\documentclass{article}
\usepackage[hidelinks]{hyperref}
\usepackage{amsmath}
\usepackage{parskip} 
\usepackage{natbib}
\usepackage{mathrsfs}
\usepackage{pdfpages}
\usepackage{bm}
\usepackage{booktabs}
\usepackage{graphics}
\usepackage{multicol}
\usepackage[bf]{caption}
\usepackage{color}
\usepackage{amssymb}
\usepackage{multirow}
\usepackage{booktabs}
\usepackage{subcaption}
\usepackage{mwe}

\usepackage{comment}
\usepackage{color}

\usepackage{amsthm}

\usepackage[margin=1.2in]{geometry}
\usepackage{floatrow}
\newfloatcommand{capbtabbox}{table}[][\FBwidth]
\usepackage{rotating}
\usepackage[titletoc]{appendix}
\usepackage[graphicx]{realboxes}
\usepackage{gensymb}
\newtheorem{lemma}{Lemma}[section]
\usepackage{siunitx}
\usepackage{setspace}
\title{
On problematic case of product approximation in Backus average
}
\author{{Filip P. Adamus}
\footnote{
Department of Earth Sciences, Memorial University of Newfoundland, Canada, {\tt adamusfp@gmail.com}}
}

\date{}

\begin{document}
\maketitle
%%%%%%%%%%%%%%%%%%%%%%
%%%%%%%%%%%%%%%%%%%%%%%%%%%%%%
%%%%%%%%%%%%%%%%%%%%%%%%%%%%%
%%%%%%%%%%
\begin{abstract}
Elastic anisotropy might be a combined effect of the intrinsic anisotropy and the anisotropy induced by thin-layering.
The Backus average, a useful mathematical tool, allows us to describe such an effect quantitatively. 
The results are meaningful only if the underlying physical assumptions are obeyed, such as static equilibrium of the material.
We focus on the only mathematical assumption of the Backus average, namely, product approximation.
It states that the average of the product of a varying function with nearly-constant function is approximately equal to the product of the averages of those functions.
We discuss particular, problematic case for which the aforementioned assumption is inaccurate.
Further, we focus on the seismological context.
We examine numerically if the inaccuracy affects the wave propagation in a homogenous medium---obtained using the Backus average---equivalent to thin layers.
We take into consideration various material symmetries, including orthotropic, cubic, and others.
We show that the problematic case of product approximation is strictly related to the negative Poisson's ratio of constituent layers. 
Therefore, we discuss the laboratory and well-log cases in which such a ratio has been noticed.
Upon thorough literature review, it occurs that examples of so-called auxetic materials (media that have negative Poisson's ratio) are not extremely rare exceptions as thought previously. 
The investigation and derivation of Poisson's ratio for materials exhibiting symmetry classes up to monoclinic become a significant part of this paper.
Except for the main objectives, we also show that the averaging of cubic layers results in an equivalent medium with tetragonal (not cubic) symmetry.
Additionally, we present concise formulations of stability conditions for low symmetry classes, such as trigonal, orthotropic, and monoclinic.
\end{abstract}
%%%%%%%%%%%%%%%%%%%%%%
\section{Introduction}
%%%%%%%%%%%%%%%%%%%%%%
The assumption of material isotropy is convenient but often inaccurate.
For instance, in the context of the elasticity of rocks, individual crystals have to be neither of the same types nor oriented randomly.
In case they are not, we encounter so-called intrinsic anisotropy.
Further, due to geological processes, the formation of rocks can be arranged in a non-random manner forming a foliated structure.
In such a situation, we consider anisotropy induced by thin layers.

The Backus average is a useful mathematical tool that provides us with a quantitative description of the anisotropy produced by thin layering~\citep{Backus}.
The isotropic layers can be replaced by the transversely-isotropic, equivalent (or, so-called, effective, or replacement) medium.
The anisotropy of such medium is a consequence of the inhomogeneity of the stack of layers only~\citep[e.g.,][Chapter 4]{SlawinskiGreen}. 
Further, as also shown by~\citet{Backus}, the transversely-isotropic constituents may be approximated by a transversely-isotropic medium, which anisotropy is a combined effect of the intrinsic anisotropy and the anisotropy induced by thin-layering~\citep{Bakulinetal}.
The Backus average can be extended to lower symmetry classes. 
We can either follow the procedure analogous to the one shown by~\citet{Backus} or use the efficient matrix formalism presented by~\citet{SchMuir}.

The equivalent medium obtained using the Backus average is a good analogy of a layered material only if the underlying assumptions of the average are satisfied.
In the literature, numerous authors dedicate their works to the assumption of the material's static equilibrium.
Among many of them are~\citet{Helbig84}, \citet{Carcione}, or~\citet{LinerFei}.
Another, but mathematical assumption introduced by~\citet{Backus} is the one of product approximation, which states that the average of the product of a rapidly-varying function with nearly-constant function is approximately equal to the product of the averages of those functions.
For more than a half-century, the researchers take the product assumption for granted.
\citet{Bosetal} are the first authors to discuss its validity in the context of the Backus average.
A year later,~\citet{BosProductApprox} find and examine statistically a particular case for which the product approximation results in spurious values.
They conclude that this problematic case is physically possible, but not likely to appear in seismology.
The aforementioned authors examine a single example of a rapidly-varying function that corresponds to the isotropic layers only.

This paper aims to continue the investigation on the particular, problematic case of product approximation.
However, we do not limit ourselves to the examples of rapidly-varying functions corresponding to isotropic layers, but we also check their analogous forms valid for anisotropic constituents.
We discuss in detail the possibility of the occurrence of inaccurate product approximation in the context of seismology. 
We relate it to the presence of negative Poisson's ratio in individual thin layers.
Rocks that exhibit such a ratio are called auxetic; in this work, we pay special attention to them.
Finally, we perform several simulations of a wave propagating in thinly-layered and equivalent media.
We compare the results to understand what is the practical influence of the problematic case of product assumption on the accuracy of the averaging process.

To be able to perform the investigation on product approximation and negative Poisson's ratio, first, we need to introduce the necessary tools and notions that we use later in the text.
Therefore, in Section~\ref{sec:background}, we discuss symmetry classes of elasticity tensors, the conditions that must be obeyed to make these tensors stable, and details of the Backus average.
Section~\ref{sec3} consists of the main body of the paper.

Can we, in any seismological scenarios, freely use the Backus average to approximate thinly layered material by the long-wave equivalent medium?
The above question has to be posed and answered, hence this paper.
%%%%%%%%%%%%%%%%%%%%%
%%%%%%%%%%%%%%%%%
\section{Theory}\label{sec:background}
%%%%%%%%%%%%%%%%%%%%%
%%%%%%%%%%%%%%%%%
%%%%%%%%%%%%%%%%%
\subsection{Symmetry classes of elasticity tensor}\label{sec:symmetries}
%%%%%%%%%%%%%%%%%%%%%%
In the theory of linear elasticity, the forces applied to a single point are expressed in terms of a stress tensor and their resultant deformations in terms of a strain tensor.
The definition of the strain tensor for infinitesimal displacements in three dimensions is
\begin{equation}\label{strain}
\varepsilon_{ij} :=\frac{1}{2}\left(\frac{\partial u_i}{\partial x_j}+\frac{\partial u_j}{\partial x_i}\right)\,
\qquad i, j \in\{1,2,3\}\,,
\end{equation} 
where subscripts $i$ and $j$\,, denote Cartesian coordinates, and $u_i$ are the components of the displacement vector describing the deformations in the $i$-th direction.
The constitutive equation relating stresses and strains is Hooke's law, namely,
\begin{equation}\label{hook}
\sigma_{ij}=\sum_{k=1}^{3}\sum_{\ell=1}^{3}c_{ijk\ell}\varepsilon_{k\ell}\,
\qquad i, j \in\{1,2,3\}\,,
\end{equation}
which states that the applied load at a point is linearly related to the deformation by the elasticity tensor, $c_{ijk\ell}$\,.
Due to the index symmetries of $c_{ijk\ell}$\,, we can replace it by $C_{mn}\,$, where $m,n\in\{1,...,6\}$\,, by following 
\begin{equation} 
\begin{cases} 
m=i\qquad {\rm if}\,\, i=j \\
n=\ell\qquad {\rm if}\,\, \ell=k
\end{cases}
\,\,\,{\rm and}\quad
\begin{cases} 
m=9-(i+j) \qquad {\rm if}\,\, i\neq j\\
n=9-(\ell+k) \qquad {\rm if}\,\, \ell\neq k
\end{cases}
.
\end{equation}
In this way, we can represent the elasticity tensor by a $6\times6$ matrix.
$C_{mn}$ can be invariant to different groups of transformations of the coordinate system.
The invariance to the orientation of the coordinate system is called material symmetry.
There are eight possible symmetry classes.
Herein, we focus on monoclinic, orthotropic, tetragonal, trigonal, transversely-isotropic (TI), cubic, and isotropic classes.

We call a tensor to be monoclinic if its symmetry group contains a reflection about a plane through the origin.
Herein, for convenience, we choose $x_3$ to be the axis along which we perform the reflection.
If we additionally rotate the coordinates by angle $\theta$ about the $x_3$-axis, where $\tan(2\theta)=2C_{45}/(C_{44}-C_{55})\,$, we can express the monoclinic tensor in its natural coordinate system~\citep[p.83]{Helbig}.
In such an orientation, the elasticity matrix has the lowest possible number of the nonzero entries~\citep[Section 5.6.3]{SlawinskiRed}.
We obtain the following stress-strain relation expressed in a matrix form,
\begin{equation}\label{stressstrain}
\left[
\begin{array}{c}
\sigma_{11}\\
\sigma_{22}\\
\sigma_{33}\\
\sigma_{23}\\
\sigma_{13}\\
\sigma_{12}\\
\end{array}
\right]
=
\left[
\begin{array}{cccccc}
C_{11}&C_{12} & C_{13} & 0 & 0 & C_{16}  \\
C_{12}&C_{22} &C_{23}& 0 & 0 & C_{26}  \\
C_{13}&C_{23}& C_{33} & 0 & 0 & C_{36}  \\
0 & 0 & 0 & C_{44} & 0 & 0 \\
0 & 0 & 0 & 0 & C_{55} & 0\\
C_{16} & C_{26}  & C_{36}  & 0 & 0 & C_{66}
\end{array}
\right]
\left[
\begin{array}{c}
\varepsilon_{11}\\
\varepsilon_{22}\\
\varepsilon_{33}\\
2\varepsilon_{23}\\
2\varepsilon_{13}\\
2\varepsilon_{12}\\
\end{array}
\right]
\,.
\end{equation}
The elasticity tensor whose symmetry group contains a two-fold, three-fold, four-fold, or $n$-fold rotation is called orthogonal, trigonal, tetragonal, or TI, respectively.
Their matrix representations having the least nonzero independent entries are the following.
\begin{equation}
\bm{C}^{\rm ort}
=
\left[
\begin{array}{cccccc}
C_{11}&C_{12} & C_{13} & 0 & 0 & 0  \\
C_{12}&C_{22} &C_{23}& 0 & 0 & 0  \\
C_{13}&C_{23}& C_{33} & 0 & 0 & 0  \\
0 & 0 & 0 & C_{44} & 0 & 0 \\
0 & 0 & 0 & 0 & C_{55} & 0\\
0 & 0 & 0  & 0 & 0 & C_{66}
\end{array}
\right]
\,,
\end{equation}

\begin{equation}
\bm{C}^{\rm trig}
=
\left[
\begin{array}{cccccc}
C_{11}&C_{12} & C_{13} & 0 & C_{15}  & 0  \\
C_{12}&C_{11} &C_{13}& 0 & -C_{15}  & 0  \\
C_{13}&C_{13}& C_{33} & 0 & 0 & 0  \\
0 & 0 & 0 & C_{44} & 0 & -C_{15}  \\
C_{15}& -C_{15}  & 0 & 0 & C_{44} & 0\\
0 & 0  & 0  & -C_{15} & 0 & \tfrac{C_{11}-C_{12}}{2}
\end{array}
\right],
\end{equation}

\begin{equation}
\bm{C}^{\rm tetr}
=
\left[
\begin{array}{cccccc}
C_{11}&C_{12} & C_{13} & 0 & 0 & 0  \\
C_{12}&C_{11} &C_{13}& 0 & 0 & 0  \\
C_{13}&C_{13}& C_{33} & 0 & 0 & 0  \\
0 & 0 & 0 & C_{44} & 0 & 0 \\
0 & 0 & 0 & 0 & C_{44} & 0\\
0 & 0 & 0  & 0 & 0 & C_{66}
\end{array}
\right]\,,
\end{equation}

\begin{equation}
\bm{C}^{\rm TI}
=
\left[
\begin{array}{cccccc}
C_{11}&C_{12} & C_{13} & 0 & 0 & 0  \\
C_{12}&C_{11} &C_{13}& 0 & 0 & 0  \\
C_{13}&C_{13}& C_{33} & 0 & 0 & 0  \\
0 & 0 & 0 & C_{44} & 0 & 0 \\
0 & 0 & 0 & 0 & C_{44} & 0\\
0 & 0 & 0  & 0 & 0 & \tfrac{C_{11}-C_{12}}{2}
\end{array}
\right].
\end{equation}

Again, we choose the $x_3$-axis to be the rotation axis. 
A cubic symmetry group contains four-fold rotations about two axes that are orthogonal to one another, whereas isotropic elasticity tensor is invariant under any rotation.
Their matrix representations are
\begin{equation}
\bm{C}^{\rm cub}
=
\left[
\begin{array}{cccccc}
C_{11}&C_{13} & C_{13} & 0 & 0 & 0  \\
C_{13}&C_{11} &C_{13}& 0 & 0 & 0  \\
C_{13}&C_{13}& C_{11} & 0 & 0 & 0  \\
0 & 0 & 0 & C_{44} & 0 & 0 \\
0 & 0 & 0 & 0 & C_{44} & 0\\
0 & 0 & 0  & 0 & 0 & C_{44}
\end{array}
\right]
\end{equation}
and
\begin{equation}
\bm{C}^{\rm iso}
=
\left[
\begin{array}{cccccc}
C_{11}&C_{11}-2C_{44} & C_{11}-2C_{44} & 0 & 0  & 0  \\
C_{11}-2C_{44}&C_{11} &C_{11}-2C_{44}& 0 & 0  & 0  \\
C_{11}-2C_{44}&C_{11}-2C_{44}& C_{11} & 0 & 0 & 0  \\
0 & 0 & 0 & C_{44} & 0 & 0  \\
0& 0  & 0 & 0 & C_{44} & 0\\
0 & 0  & 0  & 0 & 0 & C_{44}
\end{array}
\right].
\end{equation}
%%%%%%%%%%%%%%%%%
\subsection{Stability conditions for various symmetries}\label{sec:stability}
%%%%%%%%%%%%%%%%%%%%%%
The stability conditions constitute the fact that it is necessary to expand energy to deform a material.
To satisfy these conditions, a $6\times6$ matrix that represents an elasticity tensor must be positive semi-definite. 
A real symmetric matrix is positive semi-definite if and only if all its eigenvalues (or, equivalently, its principal minors) are nonnegative.
Any isotropic tensor is stable if 
\begin{equation}
C_{11}\geq\frac{4}{3}C_{44}\geq0\,.
\end{equation}
A cubic tensor must satisfy
\begin{equation}
C_{11}-C_{13}\geq0\,,\quad C_{11}+2C_{13}\geq0\,,\quad {\rm and}\quad C_{44}\geq0\,.
\end{equation}
For a TI and tetragonal tensor we require
\begin{equation}
C_{11}-|C_{12}|\geq0\,,\quad C_{33}(C_{11}+C_{12})\geq 2C_{13}^2\,,\quad C_{44}\geq0\,,\quad {\rm and}\quad C_{66}\geq0\,.
\end{equation}
The last inequality is redundant for the TI case, due to relation $2C_{66}=C_{11}-C_{12}$\,.
A trigonal tensor expressed in a natural coordinate system is stable if
\begin{equation}\label{stab:trig}
C_{11}-|C_{12}|\geq0\,,\quad C_{33}(C_{11}+C_{12})\geq 2C_{13}^2\,,\quad C_{44}\geq0\,,\quad {\rm and}\quad C_{11}-C_{12}\geq 2\frac{C_{15}^2}{C_{44}}\,.
\end{equation}
These inequalities are more complicated if a trigonal tensor is not expressed with respect to its natural coordinate system.
In such a case, not analysed herein, $C_{14}\neq0\,$.
So far we have obtained the above stability conditions by verifying the requirements for nonnegative eigenvalues. 
In the case of orthotropic and monoclinic symmetry classes, due to complicated forms of eigenvalues, we follow the nonnegative, principal-minors criterium.
For the orthotropic tensor, we get
\begin{align}
C_{11}\geq0\,,\quad C_{11}C_{22}\geq C_{12}^2\,,\quad C_{44}\geq0\,,\quad C_{55}\geq0\,,\quad C_{66}\geq0\,,\quad {\rm and} \label{cond:ort0}
\\
C_{11}C_{22}C_{33}+2C_{12}C_{13}C_{23}-C_{11}C_{23}^2-C_{22}C_{13}^2-C_{33}C_{12}^2\geq0\,. \label{cond:ort}
\end{align}
A monoclinic tensor is stable if inequalities~(\ref{cond:ort0}), (\ref{cond:ort}), and
\begin{equation}
\begin{aligned}\label{cond:mon}
C_{11}C_{22}C_{33}+2C_{12}C_{13}C_{23}-C_{11}C_{23}^2-C_{22}C_{13}^2-C_{33}C_{12}^2\geq
C_{16}^2(C_{22}C_{33}-C_{23}^2)\\
+C_{26}^2(C_{11}C_{33}-C_{13}^2)+
C_{36}^2(C_{11}C_{22}-C_{12}^2)+
2C_{16}C_{26}(C_{13}C_{23}-C_{33}C_{12})\\
+2C_{16}C_{36}(C_{12}C_{23}-C_{22}C_{13})+
2C_{26}C_{36}(C_{12}C_{13}-C_{11}C_{23})
\end{aligned}
\end{equation}
are satisfied.   
The inequalities are even more complicated for a non-natural coordinate system, where $C_{45}\neq0$\,.
We notice that for any symmetry class all the main-diagonal entries of the elasticity matrix must be nonnegative, which is simple to prove~\citep[e.g.][Exercise 4.5]{SlawinskiRed}.
Notice that the stability conditions for some of the symmetry classes are discussed in~\citet{Mouchat}.
%%%%%%%%%%%%%%%%%%%%%%
%%%%%%%%%%%%%%%%%%%%%%
%%%%%%%%%%%%%%%%%
\subsection{Backus average}
%%%%%%%%%%%%%%%%%%%%%%
The procedure of Backus averaging is based on the assumption that the averaged medium is in static equilibrium. If the top and bottom of such a medium  is subjected to the same stresses, and we set the Cartesian coordinate system in such a manner that the $x_3$-axis is vertical, then
\begin{equation}\label{const}
\sigma_{i3}\,,\quad \frac{\partial u_i}{\partial x_2}\,,\quad \frac{\partial u_i}{\partial x_1}\,,\qquad i\in\{1,\,2,\,3\}\,
\end{equation}
are vertically constant.
The remaining stresses or strains may vary significantly along the $x_3$-axis.

Physically, the above assumption is satisfied, and the Backus average makes sense, if the thickness of the averaged stack of layers, $l'$\,, is much smaller than the wavelength.
In other words, the lower wave frequency, the better accuracy of the average. 
For purposes of our numerical tests, performed in Section~\ref{sec:num}, we choose $l'$ to be at least ten times shorter than the dominant wavelength of primary wave, $\lambda_0^P$\,, which assures that the long-wave assumption is satisfied~\citep{Carcione}.

Mathematically, the Backus average is correct if the only one mathematical assumption introduced by Backus, namely, the product approximation, remains true.
As Backus states in his paper,
\begin{equation}
{\overline{f(x_3)g(x_3)}} \approx {\overline{f(x_3)}}\,\,{\overline{g(x_3)}}\,,
\end{equation}
where, overbar denotes the average weighted by the layer thicknesses.
$f(x_3)$ is a nearly-constant function that stands for stresses and displacements from expression~(\ref{const}).
$g(x_3)$ describes combinations of elasticity parameters, which can vary significantly from layer to layer.
If the above approximation holds, the elasticity coefficients of a stack of isotropic layers are long-wave equivalent to
\begin{equation}\label{eq:backus}
\begin{gathered}
    C_{11}^{\overline{\rm TI}}=\overline{\left(\frac{C_{11}-2C_{44}}{C_{11}}\right)}^2\,\overline{\left(\frac{1}{C_{11}}\right)}^{-1}+\overline{\left(\frac{4(C_{11}-C_{44})C_{44}}{C_{11}}\right)}\,,\\
%\begin{equation}
%   \label{bac2}
 %    c_{1122}^{\overline{\rm TI}}=\overline{\left(\frac{c_{1111}-2c_{2323}}{c_{1111}}\right)}^2\,\overline{\left(\frac{1}{c_{1111}}\right)}^{-1}+\overline{\left(\frac{2(c_{1111}-2c_{2323})c_{2323}}{c_{1111}}\right)}\,,
%\end{equation}
 C_{13}^{\overline{\rm TI}}=\overline{\left(\frac{C_{11}-2C_{44}}{C_{11}}\right)}\,\overline{\left(\frac{1}{C_{11}}\right)}^{-1}\,,\\
     C_{33}^{\overline{\rm TI}}=\overline{\left(\frac{1}{C_{11}}\right)}^{-1}\,,\\
     C_{44}^{\overline{\rm TI}}=\overline{\left(\frac{1}{C_{44}}\right)}^{-1}\,,\\
     C_{66}^{\overline{\rm TI}}=\overline{C_{44}}\,,
     \end{gathered}
 \end{equation}
where $C_{11}$ and $C_{44}$ describe each isotropic layer.
The five independent coefficients on the left-hand side are the equivalent transversely-isotropic parameters.
In Appendix~\ref{ch4:ap1}, we present formulation of the Backus average for layers that exhibit lower symmetry classes.
%%%%%%%%%%%%%%%%%%%%%%
%%%%%%%%%%%%%%%%%
\section{Problematic case of product approximation}\label{sec3}
%%%%%%%%%%%%%%%%%%%%%%
As discussed by~\citet{BosProductApprox}, the assumption of product approximation may be inaccurate only in the case of ${\overline{g}}\approx0$\,. 
Since, in such a situation, the relative error, 
\begin{equation}
err=\frac{{\overline{fg}}-{\overline{f}\overline{g}}}{{\overline{fg}}}\times100\%\,,
\end{equation}
is around $100\%\,$.
Predominantly, $g$ is positive~\citep{BosProductApprox}.
Therefore, in this section, we look for the possibilities of negative, or low positive $g$'s in layers so that the averaged medium have a chance to represent the problematic case of ${\overline{g}}\approx0$\,.
First, in Section~\ref{sec:th}, we study the problem from the theoretical point of view.
We analyse various examples of functions $g$ that describe combinations of elasticity coefficients corresponding to different symmetry classes.
Subsequently, in Section~\ref{sec:pois}, we look into the close relationship between Poisson's ratio and $g$\,.
Based on this relation, we discuss the possibility of occurrence of ${\overline{g}}\approx0$ in the real seismological cases.
Lastly, in Section~\ref{sec:num}, we choose theoretically and practically possible values of elasticity parameters for each layer, such that the resulting $\overline{g}\approx0$\,.
Based on numerical experiments, we compare the simulation of a wave propagating in the layered and long-wave equivalent medium.
%%%%%%%%%%%%%%%%%%%%%%
\subsection{Negative $g$}\label{sec:th}
%%%%%%%%%%%%%%%%%%%%%%
Let us examine to which combinations of elasticity parameters function $g$ corresponds.
Herein, we consider symmetry classes up to monoclinic.
To derive $g$\,, as an example, we perform the standard procedure to get Backus average for the monoclinic symmetry.
First, we write the stress-strain relations in such medium as
\begin{equation}\label{eq1}
\sigma_{11}=C_{11}\varepsilon_{11}+C_{12}\varepsilon_{22}+C_{13}\varepsilon_{33}+2C_{16}\varepsilon_{12}\,,
\end{equation}
\begin{equation}\label{eq2}
\sigma_{22}=C_{12}\varepsilon_{11}+C_{22}\varepsilon_{22}+C_{23}\varepsilon_{33}+2C_{26}\varepsilon_{12}\,,
\end{equation}
\begin{equation}\label{eq3}
\sigma_{33}=C_{13}\varepsilon_{11}+C_{23}\varepsilon_{22}+C_{33}\varepsilon_{33}+2C_{36}\varepsilon_{12}\,,
\end{equation}
\begin{equation}\label{eq4}
\sigma_{23}=C_{44}\frac{\partial{u_2}}{\partial{x_3}}+C_{44}\frac{\partial{u_3}}{\partial{x_2}}\,,
\end{equation}
\begin{equation}\label{eq5}
\sigma_{13}=C_{55}\frac{\partial{u_1}}{\partial{x_3}}+C_{55}\frac{\partial{u_3}}{\partial{x_1}}\,,
\end{equation}
\begin{equation}\label{eq6}
\sigma_{12}=C_{16}\varepsilon_{11}+C_{26}\varepsilon_{22}+C_{36}\varepsilon_{33}+2C_{66}\varepsilon_{12}\,.
\end{equation}
Then, we rewrite the above equations.
We want to have one component of a stress tensor or displacement vector that may vary along the $x_3$-axis on one side of the equations, whereas on the other side the components that are nearly constant.
We can directly do it with equations~(\ref{eq3})--(\ref{eq5}), namely,
\begin{equation}\label{eq7}
\varepsilon_{33}=\sigma_{33}\underbrace{ \left(\frac{1}{C_{33}}\right) }_{\text{$g_1$}}-\underbrace{\left(\frac{C_{13}}{C_{33}}\right)}_\text{$g_2$}\varepsilon_{11}-\underbrace{\left(\frac{C_{23}}{C_{33}}\right)}_\text{$g_3$}\varepsilon_{22}-\underbrace{\left(\frac{C_{36}}{C_{33}}\right)}_\text{$g_{m_1}$}2\varepsilon_{12}\,,
\end{equation}
\begin{equation}
\frac{\partial{u_2}}{\partial{x_3}}=\sigma_{23}\underbrace{\left(\frac{1}{C_{44}}\right)}_\text{$g_4$}-\frac{\partial{u_3}}{\partial{x_2}}\,,
\end{equation}
\begin{equation}
\frac{\partial{u_1}}{\partial{x_3}}=\sigma_{13}\underbrace{\left(\frac{1}{C_{55}}\right)}_\text{$g_5$}-\frac{\partial{u_3}}{\partial{x_1}}\,.
\end{equation}
Now, we insert the right-hand side of equation~(\ref{eq7}) into equations~(\ref{eq1}), (\ref{eq2}), and~(\ref{eq6}), so we get,
\begin{equation}
\sigma_{11}= \sigma_{33}\left(\frac{C_{13}}{C_{33}}\right)
+
\underbrace{\left(C_{11}-\frac{C_{13}^2}{C_{33}}\right)}_\text{$g_6$}\varepsilon_{11}
+
\underbrace{\left(C_{12}-\frac{C_{13}C_{23}}{C_{33}}\right)}_\text{$g_7$}\varepsilon_{22}
+
\underbrace{\left(C_{16}-\frac{C_{13}C_{36}}{C_{33}}\right)}_\text{$g_{m_2}$}2\varepsilon_{12}
\,,
\end{equation}
\begin{equation}
\sigma_{22}=
 \sigma_{33}\left(\frac{C_{23}}{C_{33}}\right)
+
\left(C_{12}-\frac{C_{13}C_{23}}{C_{33}}\right)\varepsilon_{11}
+
\underbrace{\left(C_{22}-\frac{C_{23}^2}{C_{33}}\right)}_\text{$g_8$}\varepsilon_{22}
+
\underbrace{\left(C_{26}-\frac{C_{23}C_{36}}{C_{33}}\right)}_\text{$g_{m_3}$}2\varepsilon_{12}
\,,
\end{equation}
\begin{equation}\label{eq12}
\sigma_{12}=
\sigma_{33}\left(\frac{C_{36}}{C_{33}}\right)
+
\left(C_{16}-\frac{C_{13}C_{36}}{C_{33}}\right)\varepsilon_{11}
+
\left(C_{26}-\frac{C_{23}C_{36}}{C_{33}}\right)\varepsilon_{22}
+
\underbrace{\left(C_{66}-\frac{C_{36}^2}{C_{33}}\right)}_\text{$g_9$}2\varepsilon_{12}
\,.
\end{equation}
Equations~(\ref{eq7})--(\ref{eq12}) are ready to be averaged.
However, to be able to proceed with the Backus average, from now on, we need to introduce the assumption of product approximation (see Appendix~\ref{ch4:ap1}).
Terms in parenthesis in equations~(\ref{eq7})--(\ref{eq12}) correspond to various $g$\,; we denote them as $g_i$ or $g_{m_i}\,$. 
Terms outside of parenthesis correspond to slowly varying function $f$\,.
Expressions $g_i$ are also presented in higher symmetry classes, and if we follow the procedure shown above, they occupy the analogical places as in equations~(\ref{eq7})--(\ref{eq12}).
On the other hand, $g_{m_i}$\,, are typical for a monoclinic symmetry class only; they do not have the analogical terms in higher symmetry classes.
As shown in Appendix~\ref{sec:appendix}, in trigonal symmetry there is a special case of $g$ that also does not find the analogy in other symmetries. 
We denote it by $g_{t}$\,.
In Table~\ref{tab:g}, we indicate all possibilities of $g$'s for seven symmetry classes.
\begin{table}[!htbp]
\renewcommand{\arraystretch}{1.2}
\scalebox{0.89}{%
\begin{tabular}
{ccccc}
\toprule
&monoclinic $(g^{\rm mon})$ & orthotropic $(g^{\rm ort})$& trigonal $(g^{\rm trig})$& tetragonal $(g^{\rm tetr})$\\
\cmidrule{1-5}
$g_1$& $1/C_{33}$&                                  $1/C_{33}$&        $1/C_{33}$&$1/C_{33}$\\
$g_2$& $C_{13}/C_{33}$&                         $C_{13}/C_{33}$&$C_{13}/C_{33}$&$C_{13}/C_{33}$\\
$g_3$& $C_{23}/C_{33}$&                         $C_{23}/C_{33}$&$g_2^{\rm trig}$&$g_2^{\rm tetr}$\\
$g_4$& $1/C_{44}$&                                  $1/C_{44}$&$1/C_{44}$&$1/C_{44}$\\
$g_5$& $1/C_{55}$&                                  $1/C_{55}$&$g_4^{\rm trig}$&$g_4^{\rm tetr}$\\
$g_6$& $C_{11}-C_{13}^2/C_{33}$&         $C_{11}-C_{13}^2/C_{33}$&   $C_{11}-C_{13}^2/C_{33}-C_{15}^2/C_{44}$&$C_{11}-C_{13}^2/C_{33}$\\
$g_7$& $C_{12}-C_{13}C_{23}/C_{33}$&  $C_{12}-C_{13}C_{23}/C_{33}$& $C_{12}-C_{13}^2/C_{33}+C_{15}^2/C_{44}$&$C_{12}-C_{13}^2/C_{33}$\\
$g_8$& $C_{22}-C_{23}^2/C_{33}$&         $C_{22}-C_{23}^2/C_{33}$&$g_6^{\rm trig}$&$g_6^{\rm tetr}$\\
$g_9$& $C_{66}-C_{36}^2/C_{33}$&         $C_{66}$&             $(C_{11}-C_{12})/2-C_{15}^2/C_{44}$&$C_{66}$\\
$g_{m_1}$& $C_{36}/C_{33}$&&&\\
$g_{m_2}$& $C_{16}-C_{13}C_{36}/C_{33}$& &&\\
$g_{m_3}$& $C_{26}-C_{23}C_{36}/C_{33}$& &&\\
$g_{t}$& &&$C_{15}/C_{44}$ &\\
\cmidrule{1-5}
 &TI $(g^{\rm TI})$& cubic $(g^{\rm cub})$& isotropic $(g^{\rm iso})$&\\
 \cmidrule{1-5}
 $g_1$& $1/C_{33}$&                                  $1/C_{11}$&        $1/C_{11}$&\\
$g_2=g_3$& $C_{13}/C_{33}$&                         $C_{13}/C_{11}$&$(C_{11}-2C_{44})/C_{11}$&\\
$g_4=g_5$& $1/C_{44}$&                                  $1/C_{44}$& $1/C_{44}$&\\
$g_6=g_8$& $C_{11}-C_{13}^2/C_{33}$&         $C_{11}-C_{13}^2/C_{11}$&   $4(C_{11}-C_{44})C_{44}/C_{11}$&\\
$g_7$& $C_{12}-C_{13}^2/C_{33}$&  $C_{13}-C_{13}^2/C_{11}$& $2(C_{11}-2C_{44})C_{44}/C_{11}$&\\
$g_9$& $(C_{11}-C_{12})/2$&         $C_{44}$& $C_{44}$&\\
\bottomrule
\end{tabular}
\caption{\small{Specific $g$'s for symmetry classes up to monoclinic}}
\label{tab:g}
}
\end{table}

Based on stability conditions and analysis performed below, in Table~\ref{tab:gpos}, we present for which $g$'s the negative values are allowed.
\begin{table}[!htbp]
\begin{tabular}
{ccccccc}
\toprule
monoclinic & orthotropic & trigonal& tetragonal & TI & cubic & isotropic\\
\cmidrule{1-7}
$g_2^{\rm mon}$&$g_2^{\rm ort}$&$g_2^{\rm trig}$&$g_2^{\rm tetr}$&$g_2^{\rm TI}$&$g_2^{\rm cub}$&$g_2^{\rm iso}$\\
$g_3^{\rm mon}$& $g_3^{\rm ort}$&$g_7^{\rm trig}$&$g_7^{\rm tetr}$&$g_7^{\rm TI}$& &$g_7^{\rm iso}$\\
$g_7^{\rm mon}$&$g_7^{\rm ort}$&$g_{t}$&&&&\\
$g_{m_1}$&&&&&&\\
$g_{m_2}$&&&&&&\\
$g_{m_3}$&&&&&&\\
\bottomrule
\end{tabular}
\caption{\small{Possibly negative $g$'s for symmetry classes up to monoclinic}}
\label{tab:gpos}
\end{table}
As can be easily verified numerically, the stability conditions allow $C_{13}$ and $C_{23}$ to be negative, thus, $g_2$ and $g_3$ are not necessarily positive.
Since it is required that $C_{ii}\geq 0$ (for $i\in\{1,...,6\}$)\,, we conclude that all $g_1$\,, $g_4$\,, $g_5$\,, and particular $g_9$ must be nonnegative.
Below, we analyse only the cases in which it is non-trivial to decide if $g$'s are allowed to be negative.
Since, the verdicts of possible negativity of $g$'s are obvious in cases of isotropic and cubic symmetries, let us discuss $g_6$ and $g_7$ for TI and tetragonal symmetries.
We invoke condition
\begin{equation}
C_{33}(C_{11}+C_{12})\geq2C_{13}^2\,.
\end{equation}
We know also that $C_{11}>C_{12}$\,.
From the both conditions we obtain
\begin{equation}
C_{33}C_{11}\geq C_{13}^2\,,
\end{equation}
and we infer that 
\begin{equation}
C_{33}C_{12}\geq C_{13}^2\,
\end{equation}
is not necessarily true, hence, $g_7^{\rm{TI}}$ and $g_7^{\rm{tetr}}$ may be negative, whereas $g_6^{\rm{TI}}$ and $g_6^{\rm{tetr}}$ are always nonnegative.
For trigonal symmetry, the situation is more complicated, due to parameter $C_{15}$\,.
$g_9^{\rm{trig}}$ is always nonnegative, due to condition 
\begin{equation}\label{cond:trig}
\frac{(C_{11}-C_{12})}{2}\geq\frac{C_{15}^2}{C_{44}}\,.
\end{equation}
Now we can analyse $g_6^{\rm{trig}}$\,.
We know that $C_{11}-C_{13}^2/C_{33}$ is nonnegative.
To make it negative we try to subtract something  greater or equal than $C_{15}^2/C_{44}$\,, which is $(C_{11}-C_{12})/2$\,.
We obtain
\begin{equation}
C_{11}-\frac{C_{13}^2}{C_{33}}-\frac{C_{11}-C_{12}}{2}=\frac{\frac{1}{2}C_{33}(C_{11}+C_{12})}{C_{33}}-\frac{C_{13}^2}{C_{33}}\geq0\,.
\end{equation}
Thus, $g_6^{\rm{trig}}$ must be nonnegative.
If $C_{15}=0$ then $g_7^{\rm{trig}}=g_7^{\rm{tetr}}=g_7^{\rm{TI}}$ and it means that $g_7^{\rm{trig}}$ can be negative the same as $g_7^{\rm{tetr}}$ and $g_7^{\rm{TI}}$ can. The additional stability condition~(\ref{cond:trig}) for trigonal symmetry---its other conditions are the same for TI and tetragonal symmetry---does allow it.
Also, we numerically check that $C_{15}/C_{44}$ can be negative; thus, $g_{t}$ may be negative as well.

Let us discuss the orthotropic and monoclinic case. 
Due to the complexity of inequalities~(\ref{cond:ort}) and~(\ref{cond:mon}), to decide whether particular $g$ are allowed to be negative, we perform numerical---instead of analytical---analysis only. 
Based on Monte Carlo (MC) simulations, we notice that $g_7^{\rm{ort}}$\,, $g_7^{\rm{mon}}$, $g_{m_1}\,$, $g_{m_2}\,$, or $g_{m_3}$ can be negative while the eigenvalues of the tensors are still positive.
We neither have found an orthotropic matrix with six nonnegative eigenvalues, where $g_6^{\rm{ort}}<0$ or $g_8^{\rm{ort}}<0$\,, nor monoclinic, semipositive matrix, where $g_6^{\rm{mon}}<0$\,, $g_8^{\rm{mon}}<0$\,, or $g_9^{\rm{mon}}<0$\,.
Thus, we conclude that the above $g$'s must be nonnegative and we do not include them in Table~\ref{tab:gpos}.
%%%%%%%%%%%%%%%%%
\subsection{Negative Poisson's ratio}\label{sec:pois}
%%%%%%%%%%%%%%%%%%%%
\subsubsection{Relation between $g$ and Poisson's ratio}
%%%%%%%%%%%%%%%%%%%%
In this section, we look for the alternative elastic moduli that may indicate negative $g$\,. 
Especially, we focus on the relationship between  $g<0$ and negative Poisson's ratio.
First, let us discuss the isotropic symmetry class.
To have more physical insight in possibly negative $g_2^{\rm iso}$ and $g_7^{\rm iso}$\,, we can express it in terms of Lam\'e parameters or bulk modulus and rigidity.
Knowing that $\lambda:=C_{11}-2C_{44}$ and $\mu:=C_{44}$\,, we rewrite
\begin{equation}\label{eq:iso}
g_2^{\rm iso}=\frac{\lambda}{\lambda+2\mu}=\frac{K-\frac{2}{3}\mu}{K+\frac{4}{3}\mu}\qquad {\rm and}\qquad 
g_7^{\rm iso}=\frac{2\lambda\mu}{\lambda+2\mu}
=
\frac{2(K-\frac{2}{3}\mu)\mu}{K+\frac{4}{3}\mu}
\,,
\end{equation}
where $K:=\lambda+(2/3)\mu$ denotes pure incompressibility and $\mu$ stands for sole rigidity. 
The material is stable if $\lambda\geq-(2/3)\mu$\,, $\mu\geq0$\,, and $K\geq0\,$.
Thus, the denominators of expression~(\ref{eq:iso}) must be positive. 
Therefore, $g^{\rm iso}_2$ and $g^{\rm iso}_7$ are negative if and only if
\begin{equation}
\lambda<0\qquad {\rm or}\qquad \frac{{\rm incompressibility}}{{\rm rigidity}}:=\frac{K}{\mu}<\frac{2}{3}\,.
\end{equation}
The magnitudes of $g^{\rm iso}_2$ and $g^{\rm iso}_7$ are incomparable, since $g^{\rm iso}_2$ is dimensionless, whereas $g^{\rm iso}_7$ is not.
We can express the Poisson's ratio, $\nu$\,, in terms of Lam\'e parameters, or primary and secondary waves, namely,
\begin{equation}\label{eq:iso2}
\nu_{31}:=-\frac{\varepsilon_{11}}{\varepsilon_{33}}=\frac{\lambda}{2(\lambda+\mu)}=\frac{V_P^2-2V_S^2}{2(V_P^2-V_S^2)}=\nu_{ij}\qquad i\,,j\,\in\{1,\,2,\,3\}\,.
\end{equation}
The above expression is derived for uniaxial stress in the $x_3$ direction, but is valid for any direction.
We denote the axial strain by letter $i$\,, while $j$ stands for the lateral strain.
Poisson's ratio is stable if the denominator $2(\lambda+\mu)$ is positive.
Hence, negative numerator, $\lambda$\,, implies negative $\nu$\,.
Therefore, negative Poisson's ratio is another indicator of negative $g^{\rm iso}_2$ and $g^{\rm iso}_7$\,.
Also, notice that $\nu<0$ if $V_P/V_S<\sqrt{2}$\,.
Let us discuss, the cubic symmetry.
In such a case, $g$ is negative if and only if $C_{13}$ is negative, which is tantamount to negative Poisson's ratio, since we have
\begin{equation}
\nu_{ij}=\frac{C_{13}}{C_{11}+C_{13}}\,
\end{equation}
and to satisfy the stability condition the denominator must be positive.
For TI and tetragonal symmetries, Poisson's ratio
\begin{equation}\label{eq:ti1}
\nu_{31}=\nu_{32}=\frac{C_{13}}{C_{11}+C_{12}}\,,\qquad \nu_{13}=\nu_{23}=\frac{C_{13}(C_{11}-C_{12})}{C_{33}C_{11}-C_{13}^2}\,
\end{equation}
is negative if and only if $g_2$ is negative ($C_{13}$ must be less than zero) and 
\begin{equation}\label{eq:ti2}
\nu_{21}=\nu_{12}=\frac{C_{33}C_{12}-C_{13}^2}{C_{33}C_{11}-C_{13}^2}\,
\end{equation}
is negative if and only if $g_7$ is negative.
Note that expressions~(\ref{eq:iso2}),~(\ref{eq:ti1}), and~(\ref{eq:ti2}) are also derived in~\citet{MavkoEtAl2009}.
For the trigonal symmetry class, we get the following Poisson's ratios.
\begin{gather}
\nu_{31}=\nu_{32}=\frac{C_{13}\left(C_{11}-C_{12}-2\dfrac{C_{15}^2}{C_{44}}\right)}{C_{11}\left(C_{11}-2\dfrac{C_{15}^2}{C_{44}}\right)-C_{12}\left(C_{12}+2\dfrac{C_{15}^2}{C_{44}}\right)}\,, \label{trigpois1}
\\ 
\nu_{21}=\nu_{12}=\frac{C_{12}-\dfrac{C_{13}^2}{C_{33}}+\dfrac{C_{15}^2}{C_{44}}}{C_{11}-\dfrac{C_{13}^2}{C_{33}}-\dfrac{C_{15}^2}{C_{44}}}=\frac{g_7^{\rm trig}}{g_{6}^{\rm trig}}\,, \label{trigpois2}
\\
\nu_{13}=\nu_{23}=\frac{\dfrac{C_{13}}{C_{33}}\left(C_{11}-C_{12}-2\dfrac{C_{15}^2}{C_{44}}\right)}{C_{11}-\dfrac{C_{13}^2}{C_{33}}-\dfrac{C_{15}^2}{C_{44}}}
=\frac{g_2^{\rm trig}a}{g_{6}^{\rm trig}}
\,.
\end{gather}
Let us discuss expression~(\ref{trigpois1}).
The term in the nominator in parenthesis, to which we later refer as $a$\,, must be nonnegative due to the last inequality in condition~(\ref{stab:trig}).
Thus, first parenthesis in the denominator must be also positive and equal or larger than $C_{12}$\,.
Second parenthesis must be equal or smaller than $C_{11}$\,.
Therefore, the denominator must be positive.
Due to the above analysis, we notice that $\nu_{31}$ and $\nu_{32}$ are negative if and only if $C_{13}$ is negative.
In other words, negative $g_2^{\rm trig}$ is tantamount to negative $\nu_{31}$ or $\nu_{32}$\,.
On the other hand, according to expression~(\ref{trigpois2}), negative $g_7^{\rm trig}$ is tantamount to negative $\nu_{21}$\,.
Lastly, $\nu_{13}$ and $\nu_{23}$ are negative if and only if $g_2^{\rm trig}<0$\,.

For orthotropic symmetry class, we obtain the following Poisson's ratios.
\begin{gather}
\nu_{31}=\frac{C_{13}C_{22}-C_{12}C_{23}}{C_{11}C_{22}-C_{12}^2}=\frac{n_1}{d_3}\,,\\
\nu_{32}=\frac{C_{23}C_{11}-C_{12}C_{13}}{C_{11}C_{22}-C_{12}^2}=\frac{n_2}{d_3}
\,,
\\
\nu_{21}=\frac{C_{12}C_{33}-C_{13}C_{23}}{C_{11}C_{33}-C_{13}^2}=\frac{n_3}{d_2}\,,\\
\nu_{23}=\frac{C_{23}C_{11}-C_{12}C_{13}}{C_{11}C_{33}-C_{13}^2}=\frac{n_2}{d_2}
\,,
\\
\nu_{12}=\frac{C_{12}C_{33}-C_{13}C_{23}}{C_{22}C_{33}-C_{23}^2}=\frac{n_3}{d_1}\,,\\
\nu_{13}=\frac{C_{13}C_{22}-C_{12}C_{23}}{C_{22}C_{33}-C_{23}^2}=\frac{n_1}{d_1}
\,,
\end{gather}
where denominators $d_1$\,, $d_2\,$, and $d_3$ must be positive due to the stability conditions.
Numerator $n_3=C_{33}\,g_7^{\rm ort}$\,, hence, negative $g_7^{\rm ort}$ implies negative $\nu_{21}$ and $\nu_{12}$\,.
The analysis of numerators $n_1$ and $n_2$ is more complicated since we cannot simply express it in terms of $g_i^{\rm ort}$\,.
However, based on MC simulations we notice that:
\begin{itemize}
\item{if $g_2^{\rm ort}<0$ and $g_3^{\rm ort}<0$ then $n_1$ and $n_2$ cannot be both positive,}
\item{if $g_2^{\rm ort}>0$ and $g_3^{\rm ort}>0$ then $n_1$ and $n_2$ cannot be both negative,}
\item{if $g_2^{\rm ort}>0$ and $g_3^{\rm ort}<0$ then both $n_1<0$ and $n_2>0$ are not allowed,}
\item{if $g_2^{\rm ort}<0$ and $g_3^{\rm ort}>0$ then both $n_1>0$ and $n_2<0$ are not allowed.}
\end{itemize} 
The above statements are proved analytically in Appendix~\ref{ap2}.

For a monoclinic symmetry we get \small
\begin{gather}
\nu_{31}=\frac{C_{13}C_{26}^2-C_{16}C_{23}C_{26}-C_{12}C_{26}C_{36}+C_{16}C_{22}C_{36}+C_{12}C_{23}C_{66}-C_{13}C_{22}C_{66}}{C_{66}C_{12}^2-2C_{12}C_{16}C_{26}+C_{22}C_{16}^2+C_{11}C_{26}^2-C_{11}C_{22}C_{66}}=\frac{n_1}{d_3}\,
,\\
\nu_{32}=\frac{C_{16}^2 C_{23}-C_{13}C_{16}C_{26}-C_{12}C_{16}C_{36}+C_{11}C_{26}C_{36}+C_{12}C_{13}C_{66}-C_{11}C_{23}C_{66}}{C_{66}C_{12}^2-2C_{12}C_{16}C_{26}+C_{22}C_{16}^2+C_{11}C_{26}^2-C_{11}C_{22}C_{66}}=\frac{n_2}{d_3}\,
,\\
\nu_{21}=\frac{C_{12}C_{36}^2-C_{13}C_{26}C_{36}-C_{16}C_{23}C_{36}+C_{16}C_{26}C_{33}+C_{13}C_{23}C_{66}-C_{12}C_{33}C_{66}}{C_{66}C_{13}^2-2C_{13}C_{16}C_{36}+C_{33}C_{16}^2+C_{11}C_{36}^2-C_{11}C_{33}C_{66}}=\frac{n_3}{d_2}\,
,\\
 \nu_{23}=\frac{C_{16}^2 C_{23}-C_{13}C_{16}C_{26}-C_{12}C_{16}C_{36}+C_{11}C_{26}C_{36}+C_{12}C_{13}C_{66}-C_{11}C_{23}C_{66}}{C_{66}C_{13}^2-2C_{13}C_{16}C_{36}+C_{33}C_{16}^2+C_{11}C_{36}^2-C_{11}C_{33}C_{66}}=\frac{n_2}{d_2}\,,
 \\
 \nu_{12}=\frac{C_{12}C_{36}^2-C_{13}C_{26}C_{36}-C_{16}C_{23}C_{36}+C_{16}C_{26}C_{33}+C_{13}C_{23}C_{66}-C_{12}C_{33}C_{66}}{C_{66}C_{23}^2-2C_{23}C_{26}C_{36}+C_{33}C_{26}^2+C_{22}C_{36}^2-C_{22}C_{33}C_{66}}=\frac{n_3}{d_1}\,,
 \\
 \nu_{13}=\frac{C_{13}C_{26}^2-C_{16}C_{23}C_{26}-C_{12}C_{26}C_{36}+C_{16}C_{22}C_{36}+C_{12}C_{23}C_{66}-C_{13}C_{22}C_{66}}{C_{66}C_{23}^2-2C_{23}C_{26}C_{36}+C_{33}C_{26}^2+C_{22}C_{36}^2-C_{22}C_{33}C_{66}}=\frac{n_1}{d_1}\,.
\end{gather}\normalsize
Due to complicated forms of Poisson's ratios, we again are not able to analytically express the relationship between the sign of $\nu$ and its influence on $g$\,.
Based on MC simulations, we notice that 
the negative sign of alone $g_2^{\rm mon}$\,, $g_3^{\rm mon}$\,, $g_7^{\rm mon}$\,, $g_{m_1}$\,, $g_{m_2}$\,, or alone $g_{m_3}$ cannot restrict the sign of any $\nu_{ij}$\,.
Also, all negative or all positive Poisson's ratios do not imply the negative sign of any $g$\,.
There are, however, some combinations of negative $g$ that imply negative sign of certain $\nu_{ij}$\,.
For instance, negative $g_2^{\rm mon}$\,, $g_3^{\rm mon}$\,, and $g_7^{\rm mon}$\,, or negative $g_7^{\rm mon}$\,, $g_{m_2}$\,, and $g_{m_3}$\,, or negative $g_3^{\rm mon}$\,, $g_{m_1}$\,, and $g_{m_3}$\,, imply that certain $\nu_{ij}$ must be negative.

To conclude, for isotropic, cubic, TI, and tetragonal symmetries, negative Poisson's ratio in any axial direction implies some negative $g\,$, and any negative $g$ implies some $\nu_{ij}<0$\,. 
The above is not always true for trigonal and orthotropic symmetries.
In case of a trigonal symmetry class, negative $g_t$\,, and in case of an orthotropic class, negative $g_2^{\rm ort}$ or negative $g_3^{\rm ort}$\,, do not imply that certain $\nu_{ij}<0$\,.
In monoclinic case, negative $\nu$ (in all axial directions) does not imply negative sign of any $g$\,, but some combinations of negative $g$'s imply that certain $\nu_{ij}<0$\,.
%%%%%%%%%%%%%%%%%%%%
\subsubsection{Crystals, minerals, and rocks with negative Poisson's ratio}
%%%%%%%%%%%%%%%%%%%%
Since, in the majority of symmetry classes examined by us, the presence of negative Poisson's ratio is tantamount to some negative $g$\,, it is reasonable to check the sign of this ratio for the layered rocks.
The appearance of $g<0$ in some individual layers may lead to ${\overline {g}}\approx0$ of the equivalent medium that, in turn, can cause the inaccuracy of Backus approximation.
In general, $\nu<0$ is not likely to occur in geophysical data, however, as~\citet{Zaitsev} state,
\begin{quote}
rocks with negative Poisson ratios are not rare exceptions, in contrast to conventional belief.
\end{quote}
As numerically shown by~\citet{KudelaStanoev}, negative Poisson's ratio does not occur in the global seismological case exemplified by the Preliminary reference Earth model~\citep{Dziewonski}.
However, there are some laboratory and well-log cases in which $\nu<0$ has been noticed locally.
By doing a detailed literature review, we invoke them below.

First, let us discuss naturally occurring auxetic crystals and minerals that have been investigated in a laboratory.
Using spectroscopic techniques, \citet{Cristobalite} show that $\alpha$-cristobalite, which is a low-temperature modification of a crystalline form of silica (${\rm SiO_2}$), exhibit negative Poisson's ratio.
Due to its elastic anisotropy, the value of $\nu$ varies with the direction of uniaxial stress.
Poisson's ratio occurs to range from $-0.5$ to $0.08$\,, although it remains predominantly negative.
Cristobalite widely occurs in nature~\citep{Cristobalite}.
It forms in volcanic lava domes and often can be found in acidic volcanic rocks~\citep{CristobaliteVolcanic}.
Also, it can occur in soils~\citep{CristobaliteSoil}, deep-sea cherts and porcelanites~\citep{CristobaliteCherts}, or other sedimentary rocks~\citep{CristobaliteSedimentary}.
Therefore, one should not disregard its potential influence on the rock's Poisson's ratio.  
A mineral that also may have $\nu<0$ is zeolite~(\citet{Zeolite},\,\citet{Zeolite2},\,\citet{Zeolite3}).
It naturally occurs, for instance, in deep-sea sediments or geothermal systems~\citep{ZeoliteOccur}.
More, however, very rare auxetic minerals are indicated by~\citet{Baughman}.
Also, as stated by~\citet{Lakes}, it is more likely for the highly anisotropic minerals or crystals to have negative Poisson's ratio, than for the isotropic ones.
For instance, single-crystal form of anisotropic arsenic, antimony, and bismuth exhibit $\nu<0$ in certain directions.
However, there is a case of an auxetic mineral that is isotropic. 
It occurs that, depending on the temperature, polycrystalline quartz exhibits low, very low, or negative Poisson's ratio~\citep{Quartz}.
According to to~\citet{JiEtAl2010}, the presence of this mineral may render some rocks to be auxetic.
At ambient conditions, the Poisson's ratio of quartz is $\nu=0.08$\,~\citep{JiEtAl2018}.

Let us invoke some laboratory examples of various auxetic rocks.
\citet{Nur} have noticed that very small or negative values of $\nu$ are exhibited by dry rocks at very low pressure. 
They observed that, if there is no external pressure, Casco and Westerly granites present $\nu=-0.100$ and $\nu=-0.094\,$, respectively. 
Twenty years later,~\citet{Houpert} examine Senones and Remiremont granites with thermally induced cracks.
These rocks occur to have negative Poisson's ratio for various directions of uniaxial stresses, which is probably caused, as authors state, by numerous microcracks. 
The investigation of~\citet{Zaitsev} confirm that negative $\nu$ can primarily occur in cracked rocks at low pressures.
Based on the works of~\citet{Coyner},~\citet{Freund}, and~\citet{MavkoJizba}, they invoke thirty-four rock samples of cracked rocks with $\nu<0$ at $8\,\rm{MPa}$ confining pressure.
\citet{Gregory} examined twenty samples of sedimentary rocks at different pressures and at ambient temperature.
He notices that apart from the low pressure, $\nu<0$ (presented in many examined samples) is caused by gas saturation and low porosity.
The above statement is confirmed by the experiments of~\citet{Han} and~\citet{Jizba}.
In their works, negative Poisson's ratio is exhibited only by low-porous sedimentary rocks; consolidated sandstones and gas sandstones, respectively.
Such results were obtained, \textit{inter alia}, for the approximate effective pressure in the well, that is, for $20\,\rm{MPa}$~\citep{Dvorkin}.
\citet{JiEtAl2010} show that $\nu$ decreases with increasing temperature due to thermal effects.
According to the authors, quartz-rich rocks at a temperature approaching 
the $\alpha$--$\beta$ quartz transition (such as granite, diorite, quartz-rich sandstone, etc.) may display negative values of Poisson's ratio.
They use quartzite as an example to illustrate the effect of phase transition on $\nu$\,.
The quartz-transition temperature is at about \SI{600}{\celsius}, however, quartzite has $\nu=0$ at the temperature of \SI{450}{\celsius} only.
Between \SI{450}{\celsius} and \SI{600}{\celsius} it exhibits $\nu<0$~\citep[Figure 3b]{JiEtAl2010}.
Recently, the topic of auxetic natural rocks has been studied carefully by~\citet{JiEtAl2018}. They state that
\begin{quote}
none of the crystalline igneous and metamorphic rocks (e.g., amphibolite, gabbro, granite, peridotite, and schist) display auxetic behaviour at pressures of $>5\,\rm{MPa}$ and room temperature. Our experimental measurements showed that quartz-rich sedimentary rocks (i.e., sandstone and siltstone) are most likely to be the only rocks with negative Poisson's ratios at low confining pressures ($\leq200\,\rm{MPa}$) because their main constituent mineral, $\alpha$-quartz, already has extremely low Poisson's ratio ($\nu=0.08$) and they contain microcracks, micropores, and secondary minerals.
\end{quote}
In the most recent work on auxetic rocks,~\citet{JiEtAl2019} state that
\begin{quote}
negative Poisson's ratio cannot occur in wet volcanic rocks but may appear in a dry basalt with such an extremely high porosity $(\geq70\%)$ that a re-entrant foam structure has formed.
\end{quote}
Hence, apart from the laboratory experiments, we expect to detect the negative $\nu$ in the seismological studies, in the quartz-rich continental crust with a high geothermal gradient~\citep{JiEtAl2010}, quartz-rich and gas-bearing sedimentary rocks, or in dry, highly porous basalts.

Finally, we invoke the examples of auxetic rocks obtained from well-log measurements.
Let us consider the work of~\citet{CastagnaSmith}, where the worldwide collection of twenty--five sets of velocity and density measurements is exhibited.  
These measurements of brine sands, shales and gas sands, are based on well-log and laboratory data and occur in close in-situ proximity.
Based on the velocities of primary and secondary waves, along with the densities, we compute $\nu$\,.
Two samples of gas sands occur to have negative Poisson's ratio, whereas
another sample has a positive value, but very close to zero.
Their values are $\nu=-0.18$\,, $\nu=-0.0162$\,, and $\nu=1.02\times 10^{-4}$\,, respectively.
We find another example of well-log data with negative Poisson's ratio in~\citet[Table 2]{Goodway}.
The ostracod shale from the Mannville Group in Western Canadian Sedimentary Basin (WCSB), occurs to be auxetic.
The ostracod beds are used in the gas and oil exploration~\citep{HayesEtAl1994}, and in particular, oils are sourced from the ostracod shales~\citep{FayEtAl2012}.
Hence, this is an important case from the explorational point of view.
Based on the density, along with the P and S velocities, we again compute Poisson's ratio and obtain $\nu=-0.11$\,.
Further,~\citet{Crewes} present a substantial collection of data from twelve wells in offshore Newfoundland, Eastern Canada.
Based on the P and S wave velocities from their Figure 5, we read that subsets of a dataset from at least two wells present $\nu<0$\,. 

The above real-data examples confirm most of the expectations coming from the laboratory measurements. 
To conclude, the ideal conditions for a rock to be auxetic are high temperature and low pressure.
Additionally, the chances for the auxetic behaviour are larger if the rock is dry or gas-bearing, quartz-rich, has numerous cracks, and low porosity. 
%%%%%%%%%%%%%%%%%%%%%%
%%%%%%%%%%%%%%%%%%%%%%
\subsection{Numerical examples}\label{sec:num}
%%%%%%%%%%%%%%%%%%%%%
Let us consider some numerical examples to check if the signal that propagates through thin layers would change its shape and magnitude if propagating through the equivalent medium with ${\overline{g}}\approx0$\,.
In cases that we examine, Poisson's ratio of each layer is low. 
In turn, $g$'s of individual constituents are close to zero, which causes the average ${\overline{g}}\approx0$\,.
We use some practical examples of $\nu$ from Section~\ref{sec:pois}.
%%%%%%%%%%%%%%%%%%%%%
\subsubsection{Wave propagation modelling}
%%%%%%%%%%%%%%%%%%%%%
In this paper, we analyse the wave propagation in two dimensions, namely, in the $x_3x_1$-plane.
In such a case the elastic equations of motion have the following form.
\begin{equation}
\rho\frac{\partial^2 u_i}{\partial t^2}=\frac{\partial \sigma_{i1}}{\partial x_1}+\frac{\partial \sigma_{i3}}{\partial x_3}+f_i\,, \qquad i\in\{1,\,3\}\,,
\end{equation}
where $\rho$ is a mass density and $f$ is a body force.
To obtain the wave equations, we need to insert---into the equations of motion above---the 2D stress-strain relations and expression~(\ref{strain}).
To do so, we first reduce relations~(\ref{stressstrain}) to two dimensions, namely,
\begin{equation}\label{stressstrain2}
\left[
\begin{array}{c}
\sigma_{11}\\
\sigma_{33}\\
\sigma_{13}\\
\end{array}
\right]
=
\left[
\begin{array}{cccccc}
C_{11}& C_{13} & 0   \\
C_{13}& C_{33} & 0  \\
 0 & 0 & C_{55} \\
\end{array}
\right]
\left[
\begin{array}{c}
\varepsilon_{11}\\
\varepsilon_{33}\\
2\varepsilon_{13}\\
\end{array}
\right]
\,,
\end{equation}
which are the strain-stress relations valid not only for monoclinic, but also orthogonal, tetragonal, and TI symmetry classes.
For cubic case, $C_{33}=C_{11}$\,, whereas for isotropic class of symmetry, additionally, $C_{13}=C_{11}-2C_{55}$\,. 
(Herein, we do not consider a trigonal class).
We solve the resulting elastic wave equations using the open-source seismic modelling code {\textit{ewefd2d}} in the Madagascar package~\citep{Fomel}.
The code implements a time-domain finite difference method.

To be able to solve the wave equations numerically, we need to define a computational mesh.
We model seismic data on a $N_{x_1}\times N_{x_3}=1500^2$ mesh at uniform $\Delta x_1=\Delta x_3=2 {\, \rm m}$ spacing.
We assume low-frequency stress source injected in the $x_3$-axis only.
For this purpose, we use the Ricker wavelet with a dominant frequency of $12\,\rm{Hz}$\,.
We locate the source at $(x_1,\,x_3)=(1500\,{\rm m},\,1500\,{\rm m})$ and a receiver at $(x_1,\,x_3)=(1500\,{\rm m},\,1620\,{\rm m})$.

In our simulations, we consider a periodic, three-layered system, to which we refer as a PL medium.
We propose five different examples of PL media that we denote by roman letters.
Media $I$--$III$ represent isotropic layers. 
Medium $IV$ consists of cubic layers, whereas $V$ is composed of layers that exhibit either monoclinic, orthotropic, tetragonal or TI symmetry classes.
Layers are placed horizontally and uniformly separated by $2\Delta x_3=4\,\rm{m}\,$.
Hence, the receiver is separated from the source by a PL medium that consists of ten sets of three layers (in total $120\,\rm{m}$).
Elasticity coefficients and densities of the layers are presented in Table~\ref{tab:models}.
\begin{table}[!htbp]
\scalebox{0.9}{%
\begin{tabular}
{cccccc}
\toprule
 & $I \,({\rm iso})$ & $II\, ({\rm iso})$ & $III \,({\rm iso})$ & $IV\,({\rm cubic})$ & $V\, ({\rm mon/ort/tetr/TI})$\\
\cmidrule{1-6}
layer $1$ &$C_{11}=37.79$&$C_{11}=37.79$&$C_{11}=40$&$C_{11}=45\quad C_{55}=10$&$C_{11}=45\quad C_{55}=10$\\
&$C_{55}=18.89$&$C_{55}=18.89$&$C_{55}=20$&$C_{13}=1.2\times10^{-7}$&$C_{13}=1.2\times10^{-7}$\\
&$\rho=2410$&$\rho=2410$&$\rho=2410$&$\rho=2200$&$C_{33}=35 \quad \rho=2200$\\
\cmidrule{1-6}
layer $2$ &$C_{11}=5.93$&$C_{11}=20.29$&$C_{11}=20$&$C_{11}=20\quad C_{55}=5$&$C_{11}=20\quad C_{55}=5$\\
&$C_{55}=2.78$&$C_{55}=10.14$&$C_{55}=10$&$C_{13}=1.0\times10^{-7}$&$C_{13}=1.0\times10^{-7}$\\
&$\rho=2100$&$\rho=2300$&$\rho=2300$&$\rho=1800$&$C_{33}=15\quad \rho=1800$\\
\cmidrule{1-6}
layer $3$ &$C_{11}=62.44$&$C_{11}=37.79$&$C_{11}=40$&$C_{11}=30\quad C_{55}=8$&$C_{11}=30\quad C_{55}=8$\\
&$C_{55}=28.21$&$C_{55}=18.89$&$C_{55}=20$&$C_{13}=0.8\times10^{-7}$&$C_{13}=0.8\times10^{-7}$\\
&$\rho=2590$&$\rho=2410$&$\rho=2410$&$\rho=2000$&$C_{33}=22\quad \rho=2000$\\
\bottomrule
\end{tabular}
\caption[\small{Elastic properties of periodic three-layered (PL) media}]{\small{Five different elastic PL media. Elasticity parameters are in $\rm{GPa}$, whereas density in ${\rm kg/m}^3\,$.}}
\label{tab:models}
}
\end{table}

Also, using expression~(\ref{eq:backus}) and formulas from Appendix~\ref{ch4:ap1}, we compute the elastic coefficients of the media equivalent to $I$--$V$.
The equivalent density is the arithmetic average of densities of individual layers. 
To model the wave propagation---similarly to the PL case---we insert the computed parameters of the equivalent media into the wave equations.

We compare the displacement propagation in PL and equivalent media by using the following semblance.
\begin{equation}
S=\frac{\sum_i(a_i+b_i)^2}{2\sum_i(a_i^2+b_i^2)}\times 100\%\,,
\end{equation}
where $a_i$ and $b_i$ are the discrete values of displacement changing with time in both media.
%%%%%%%%%
\subsubsection{Results}
%%%%%%%%%
Medium $I$ consists of isotropic layers corresponding to gas-bearing sandstones presented in~\citet{CastagnaSmith} (sets 6, 15, and 12, respectively).
Poisson's ratio of each layer is low, namely, $\nu_1\approx2.6\times10^{-4}$\,, $\nu_2\approx0.06$\,, and $\nu_3\approx0.10$\,.
As a result, the averaged ${\overline{g_2}}\approx0.05$ is low as well.
We present the propagation of displacement ($x_3$-component) recorded by the receiver in Figure~\ref{fig:rec1}.
In Figure~\ref{fig:disp1}, we additionally show the snapshots of wave propagation in both components recorded at time $t=0.3\,\rm{s}\,$. 
We notice that displacements are almost identical for both PL and equivalent media, which is confirmed by $S\approx99.99\%$\,. 
Hence, the product approximation, even if $\overline{g}$ is low, seems to be correct, and the average works properly. 
%%%
\begin{figure}[!htbp]
        \centering
      \begin{subfigure}[b]{0.45\textwidth}   
            \centering 
            \includegraphics[width=\textwidth]{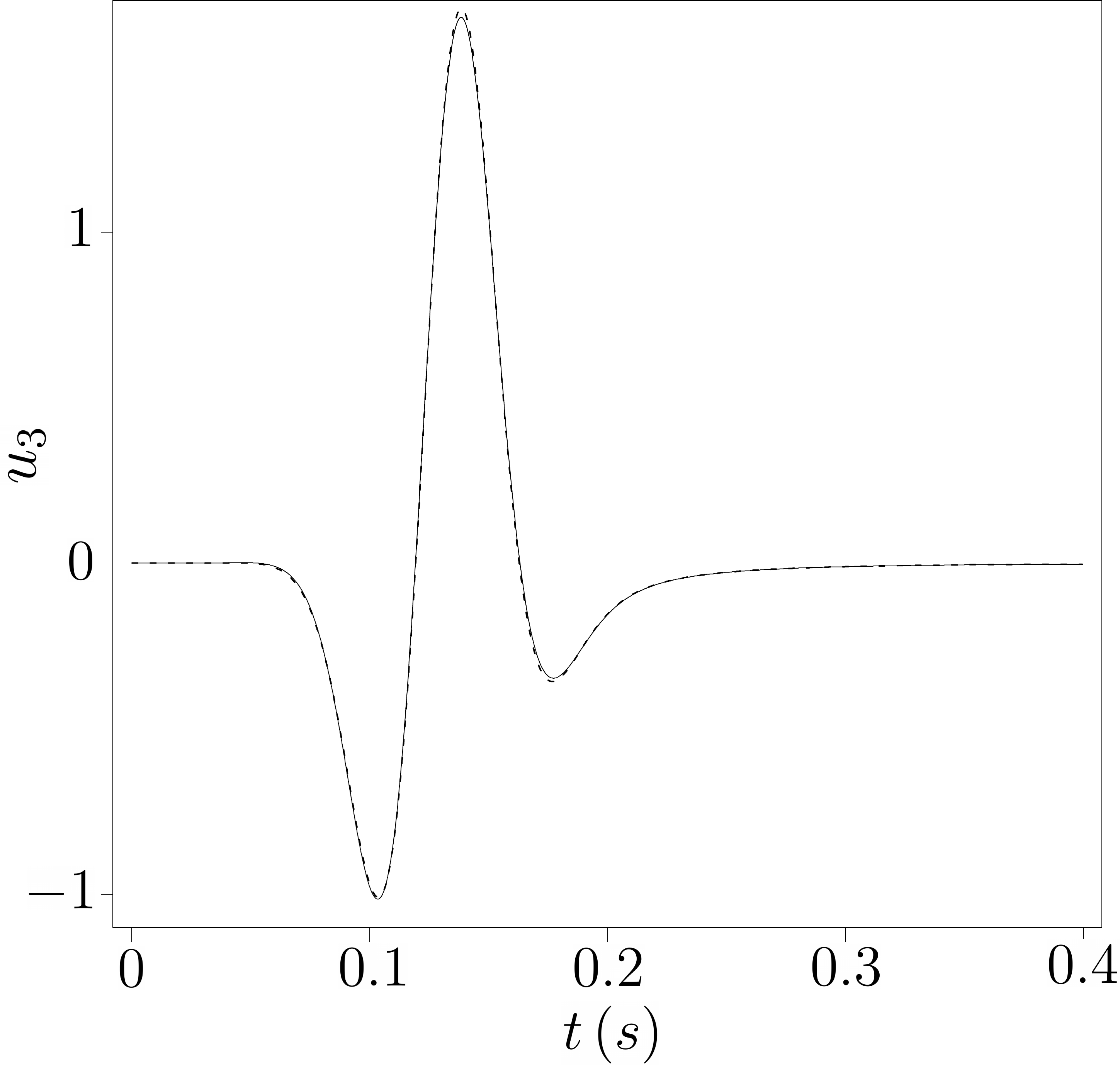}
            \caption[]%
            {{\footnotesize $12\,\rm{Hz}$ dominant frequency}}    
            \label{fig:rec1}
        \end{subfigure}
        \qquad
        \begin{subfigure}[b]{0.45\textwidth}   
            \centering 
            \includegraphics[width=\textwidth]{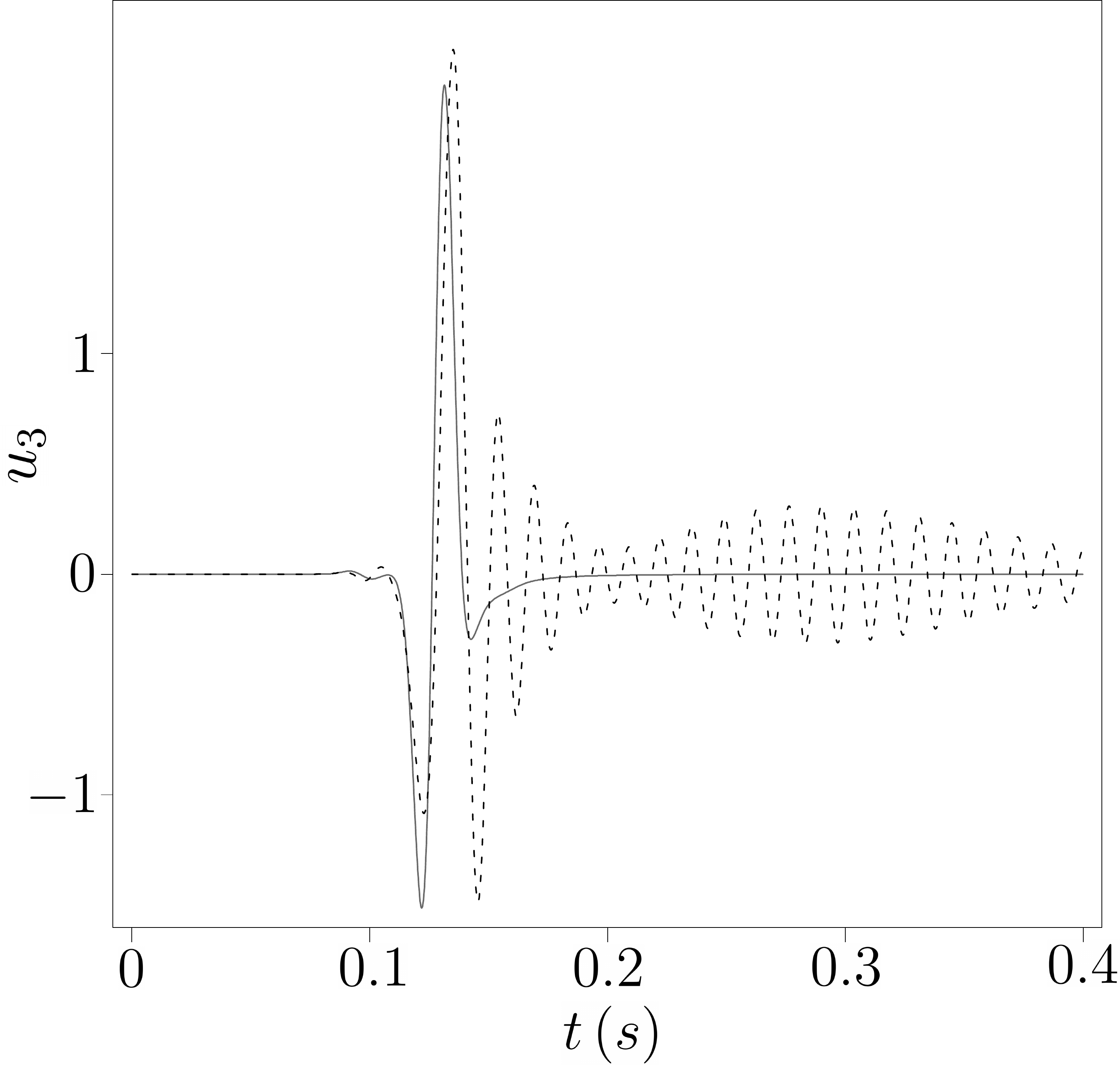}
            \caption[]%
            {{ \footnotesize $48\,\rm{Hz}$ dominant frequency}}    
            \label{fig:rec1s}
        \end{subfigure}
        \caption[\small{Displacement $u_3$ recorded by the receiver: PL and equivalent medium $I$}]
        {\small Displacement $u_3$ recorded by the receiver. Signal in PL and equivalent medium $I$ is denoted by dashed and solid line, respectively.} 
        \label{fig:rec}
    \end{figure}

The properties of Medium $II$ are similar to gas sandstone from~\citet{CastagnaSmith} (layer $1$ and $3$) and ostracod shale from~\citet{Goodway} (layer $2$).
In this example, however, we choose the elasticity parameters in such a way that the resulting ${\overline{g_2}}\approx3\times10^{-4}$ is very low, but still possible to occur in real data case.
The semblance of displacement propagation recorded by the receiver in PL and equivalent medium is high, $S\approx100\%$\,.
Again, Backus approximation appears to be accurate.

Medium $III$ is the idealised version of the previous example. We slightly change the elasticity parameters in a way that ${\overline{g_2}}$ is precisely zero, which is probably impossible to achieve in real seismological case.
Thus, in this example, the relative error of the product approximation is $100\%$\,.
Perhaps surprisingly, the aforementioned error does not influence the accuracy of the Backus average, since $S\approx100\%$\,.

Previous examples regarded isotropic layers only. 
From now on, however, we focus on anisotropic constituents.
Medium $IV$ presents cubic layering. 
The medium equivalent to cubic layers has a tetragonal symmetry class, as we show in Appendix~\ref{ch4:ap1}.
This fact might not be evident for the readers since we have not encountered the above statement or analogical examples in the existing literature.
We set $C_{13}$ to have very small values, so that ${\overline{g_2}}\approx3\times10^{-9}$ appears to be extremely low.
As in previous examples, Backus average works properly ($S\approx100\%$), which is illustrated by Figures~\ref{fig:rec4} and~\ref{fig:disp4}.
The last Medium $V$ represents layers that can exhibit monoclinic, orthotropic, tetragonal, or TI symmetry class. 
The product approximation is inaccurate due to ${\overline{g_2}}\approx5\times10^{-9}$\,.
However, again it does not affect the Backus approximation that is still accurate since $S\approx100\%$\,.
\newpage
It occurs that the low-frequency assumption seems to raise more concerns than the product assumption.
To support the above statement, let us again consider Medium $I$\,, but exceptionally change the dominant frequency of the Ricker wavelet to $48\,\rm{Hz}$\,; thus, let us increase it four times.

Figures~\ref{fig:rec1s} and~\ref{fig:disp1s} illustrate the inaccuracy of the averaging process confirmed by $S\approx82.81\%$ only. 
Later in the text, we refer to this higher-frequency example as to the case $I^*$\,.
For reference, in Table~\ref{tab:semblance}, we present more accurate values of semblances and ${\overline{g_2}}$ for all cases $I$--$V$\,.
%%%
\begin{figure}[!]
        \centering
        \begin{subfigure}[b]{0.45\textwidth}
            \centering
            \includegraphics[width=\textwidth]{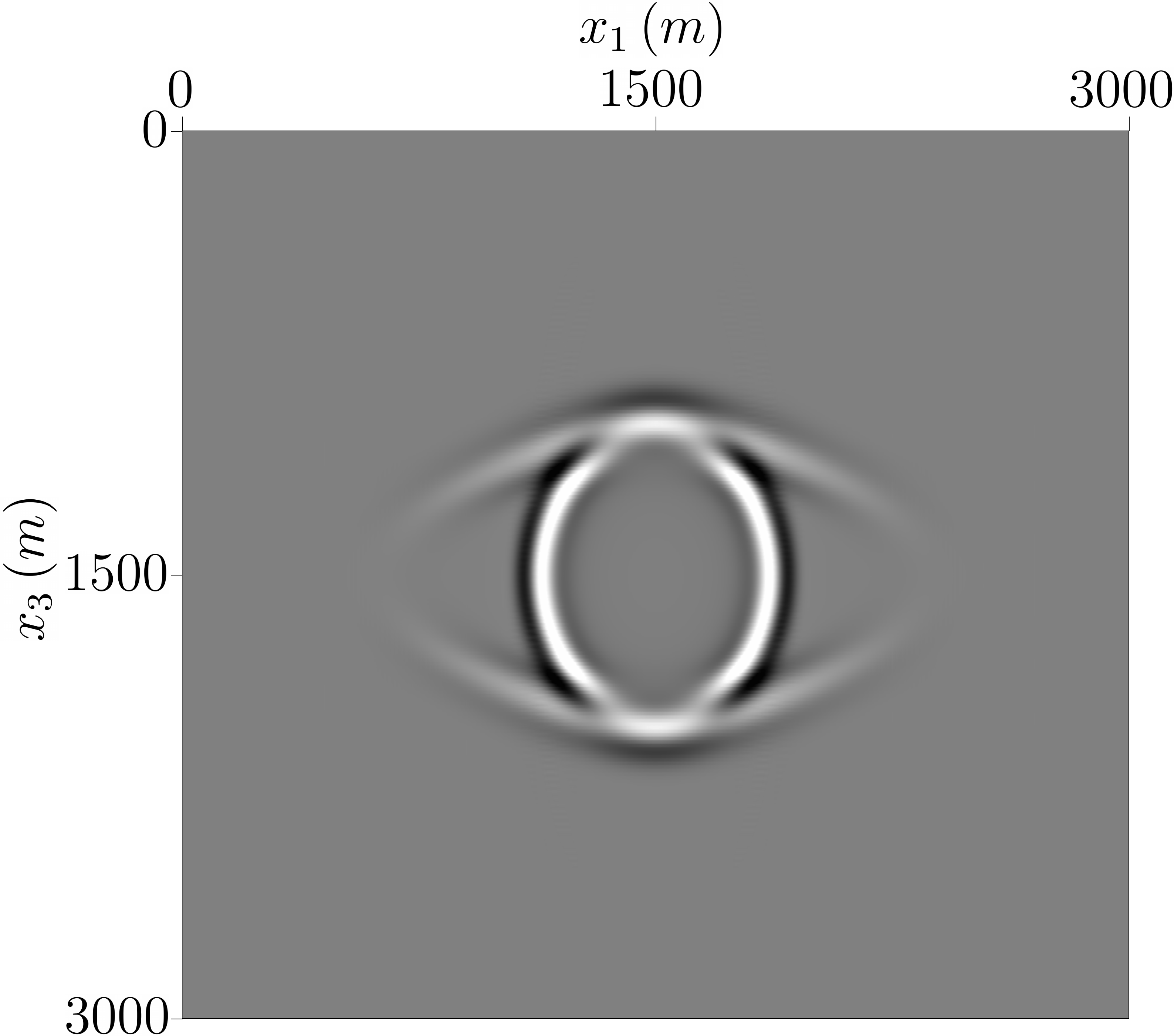}
            \caption%
            {{\footnotesize $u_3$ in PL medium}}    
            \label{fig:disp1a}
        \end{subfigure}
        \quad
        \begin{subfigure}[b]{0.45\textwidth}  
            \centering 
            \includegraphics[width=\textwidth]{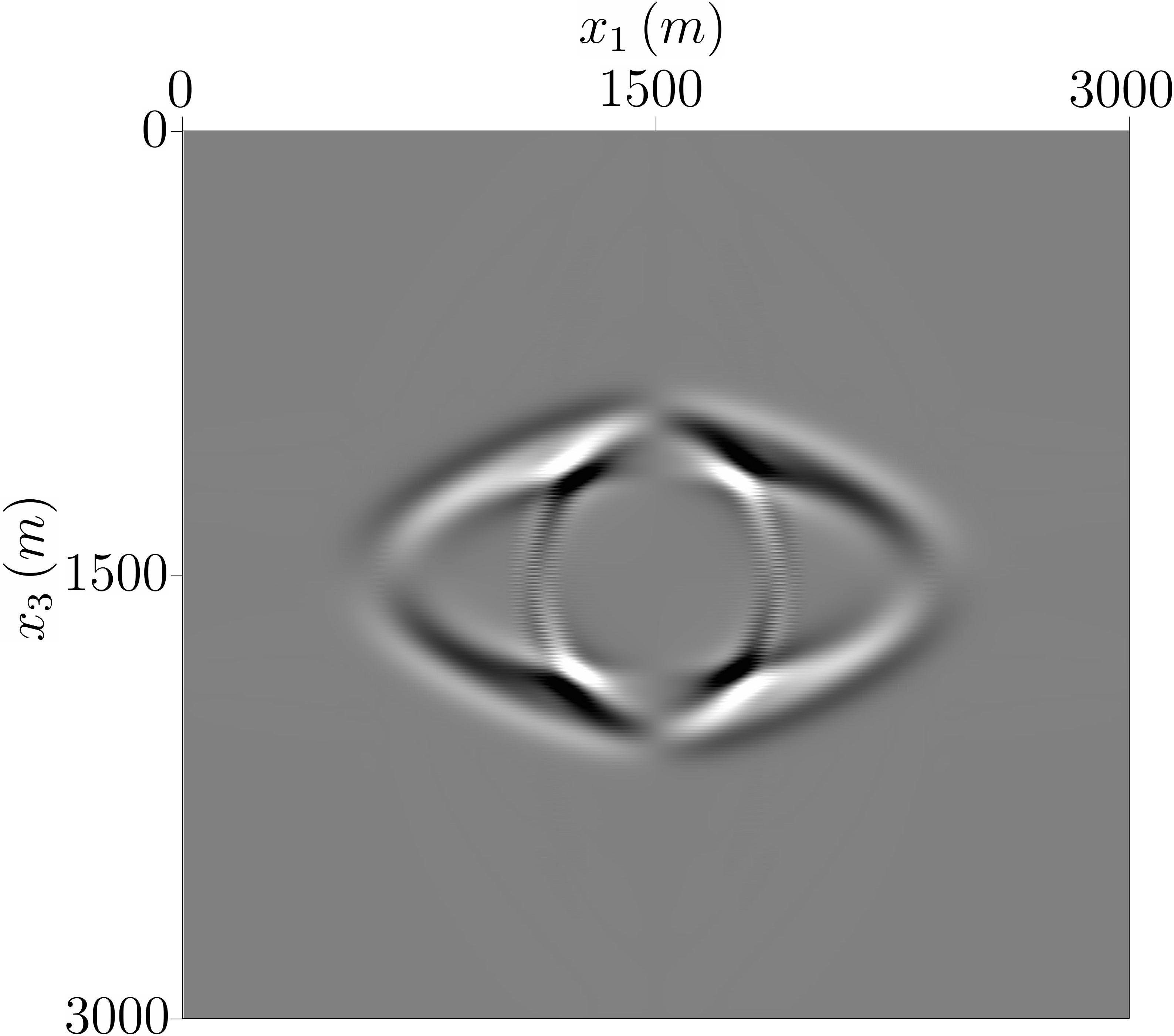}
            \caption[]%
            {{\footnotesize $u_1$ in PL medium}}    
            \label{fig:disp1b}
        \end{subfigure}
        \vskip\baselineskip
        \begin{subfigure}[b]{0.45\textwidth}   
            \centering 
            \includegraphics[width=\textwidth]{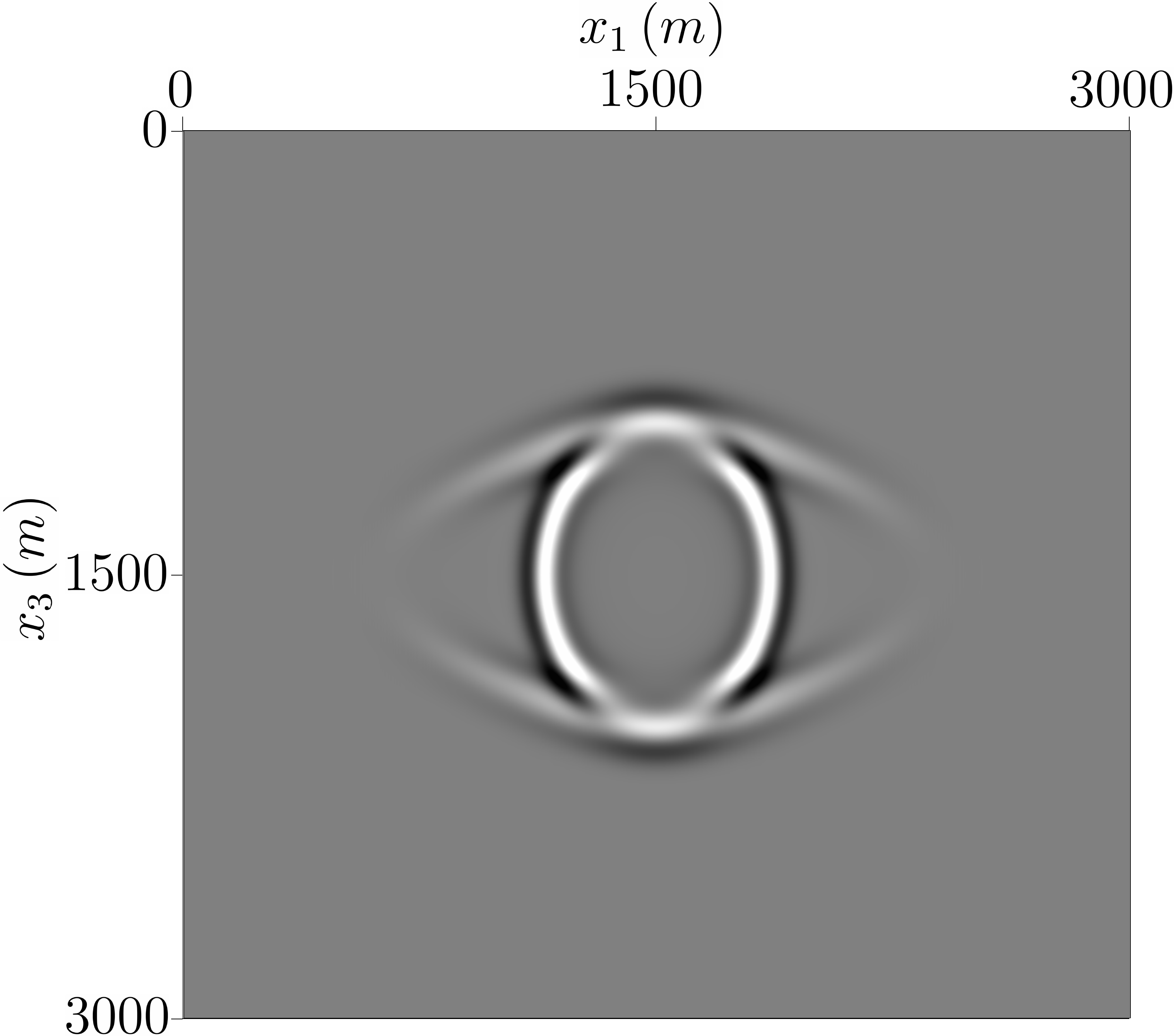}
            \caption[]%
            {{\footnotesize $u_3$ in equivalent medium}}    
            \label{fig:disp1c}
        \end{subfigure}
        \quad
        \begin{subfigure}[b]{0.45\textwidth}   
            \centering 
            \includegraphics[width=\textwidth]{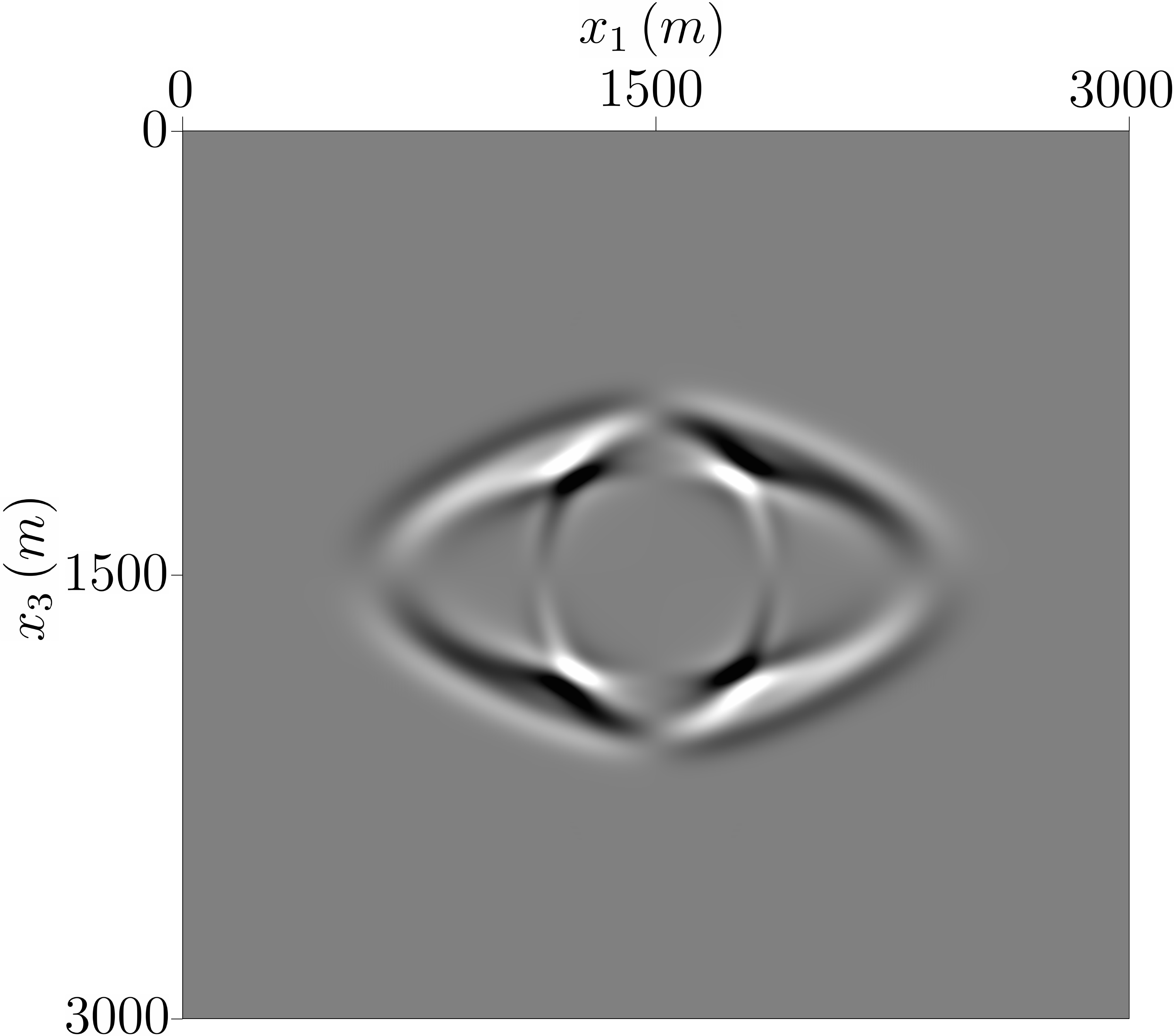}
            \caption[]%
            {{ \footnotesize $u_1$ in equivalent medium}}    
            \label{fig:disp1d}
        \end{subfigure}
        \caption[\small{Snapshots of displacement $u_3$ and $u_1$\,: PL and equivalent medium $I$}]
        {\small Snapshots of displacement $u_3$ and $u_1$ in PL and equivalent medium $I$ at time $t=0.3\,\rm{s}$} 
        \label{fig:disp1}
    \end{figure}
%%%%
%%%%%%%%%%%%%%%%%%%%%
\begin{figure}[!htbp]
        \centering
        \begin{subfigure}[b]{0.45\textwidth}
            \centering
            \includegraphics[width=\textwidth]{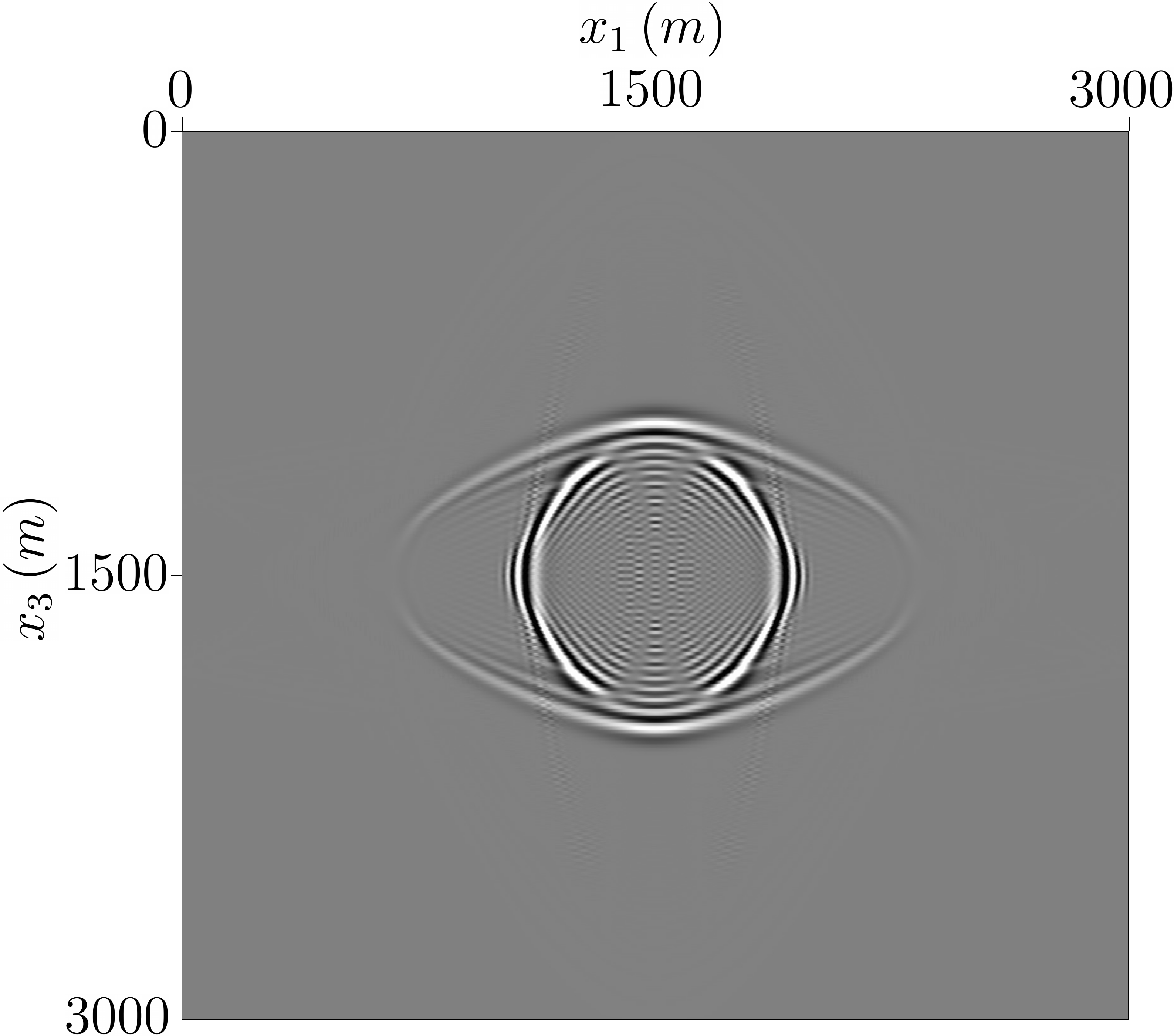}
            \caption%
            {{\footnotesize $u_3$ in PL medium}}    
            \label{fig:disp1sa}
        \end{subfigure}
        \quad
        \begin{subfigure}[b]{0.45\textwidth}  
            \centering 
            \includegraphics[width=\textwidth]{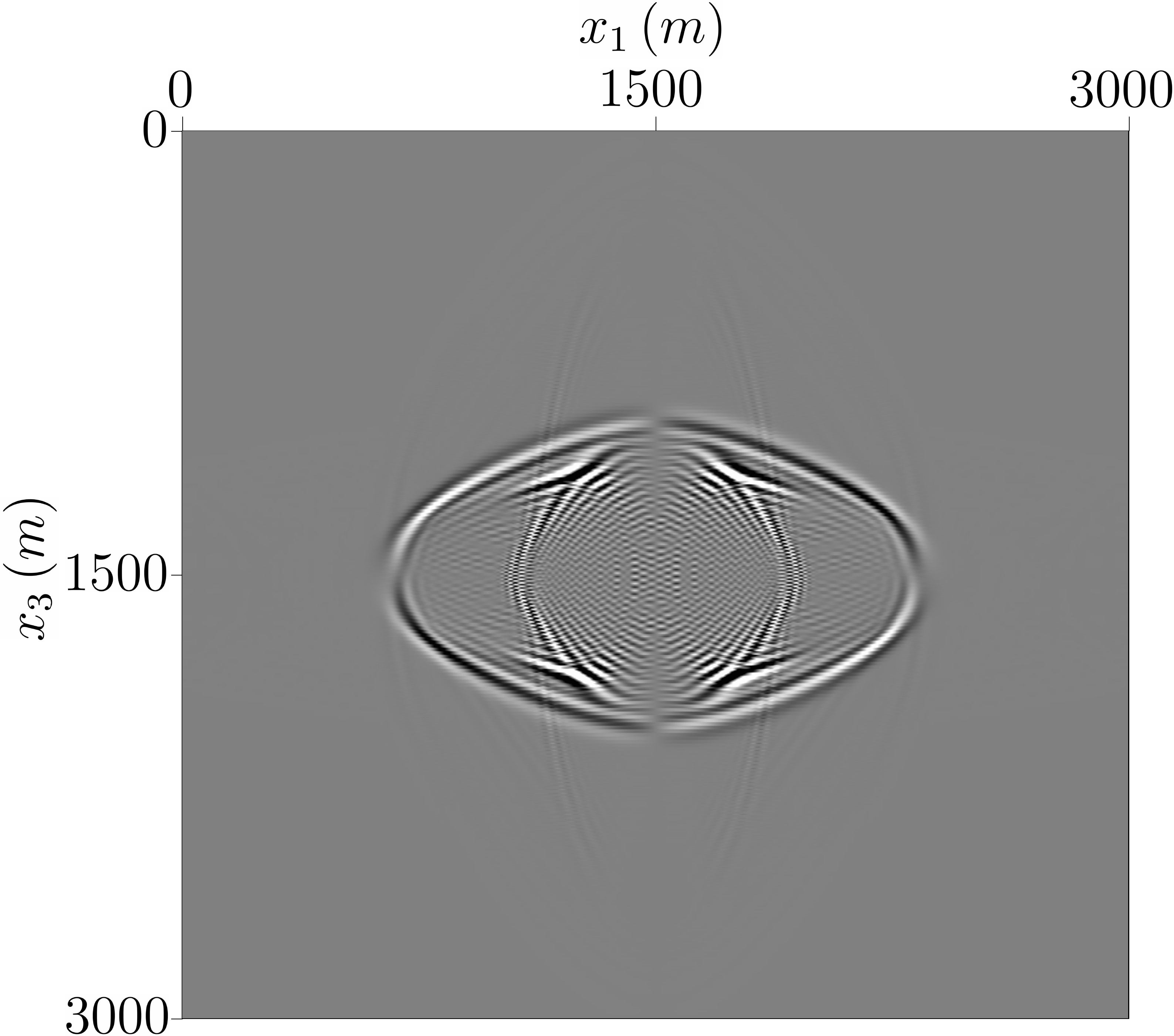}
            \caption[]%
            {{\footnotesize $u_1$ in PL medium}}    
            \label{fig:disp1sb}
        \end{subfigure}
        \vskip\baselineskip
        \begin{subfigure}[b]{0.45\textwidth}   
            \centering 
            \includegraphics[width=\textwidth]{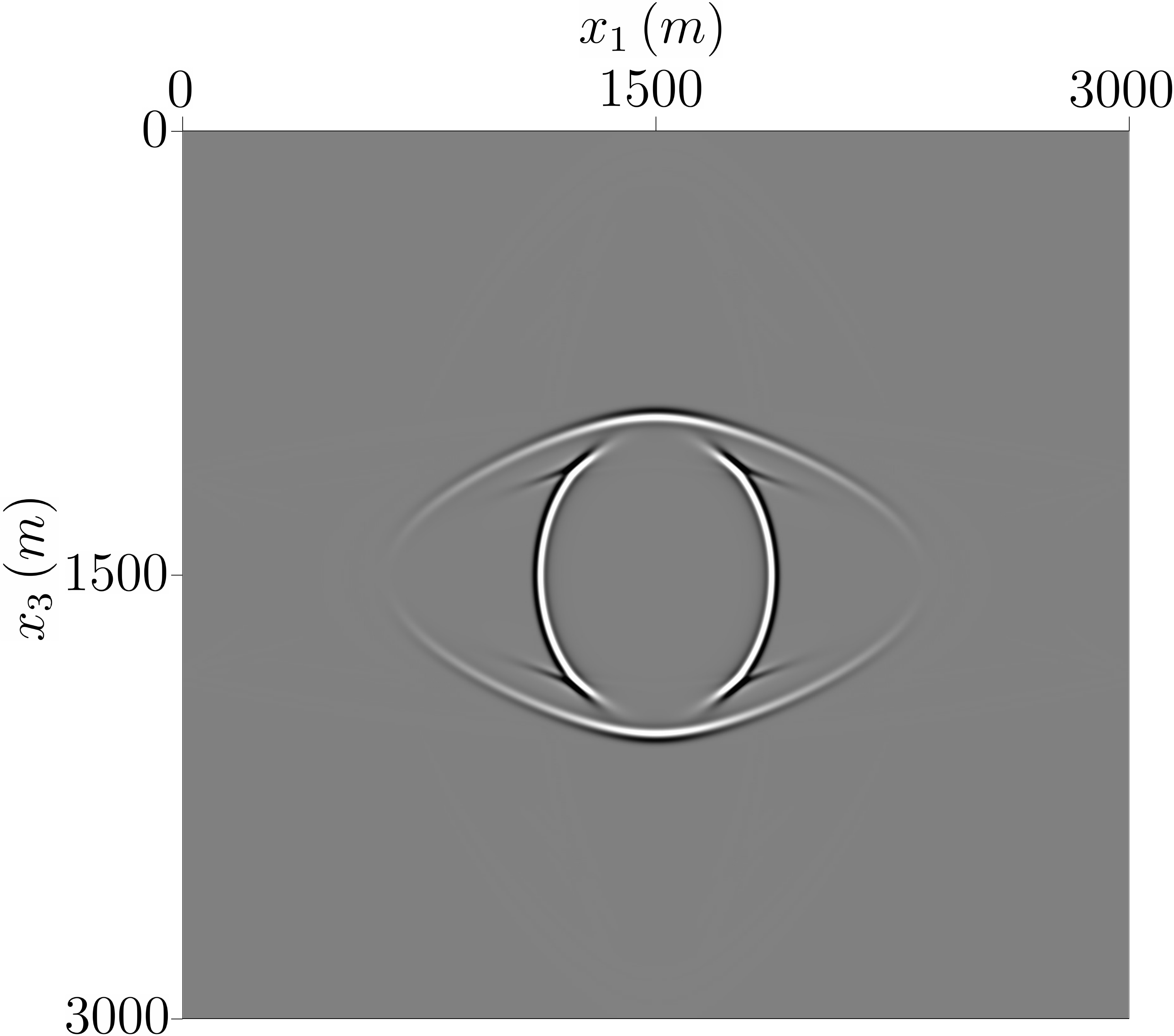}
            \caption[]%
            {{\footnotesize $u_3$ in equivalent medium}}    
            \label{fig:disp1sc}
        \end{subfigure}
        \quad
        \begin{subfigure}[b]{0.45\textwidth}   
            \centering 
            \includegraphics[width=\textwidth]{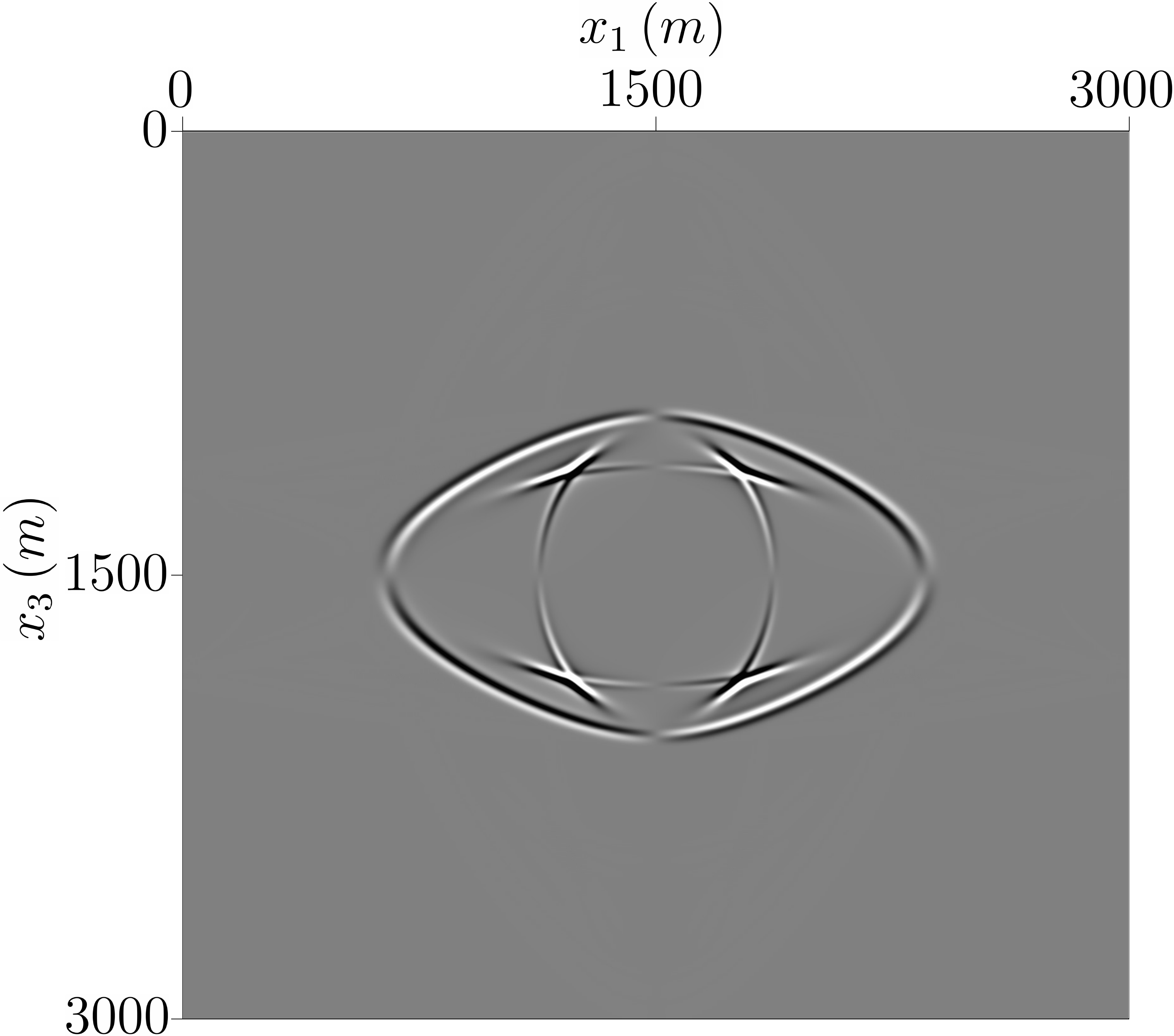}
            \caption[]%
            {{\small \footnotesize $u_1$ in equivalent medium}}    
            \label{fig:disp1sd}
        \end{subfigure}
        \caption[\small{Snapshots of displacement $u_3$ and $u_1$\,: PL and equivalent medium $I^*$}]
        {\small Snapshots of displacement $u_3$ and $u_1$ in PL and equivalent medium $I^*$ at time $t=0.3\, \rm{s}$} 
        \label{fig:disp1s}
    \end{figure}
%%%%%%%%%%%%%%%%%%%%%

%%%%%%%%%%%%%%%
    \begin{table}[!htbp]
    \renewcommand{\arraystretch}{1.1}
\begin{tabular}
{ccccccc}
\toprule
 & $I$ & $I^*$ & $II$ & $III$ & $IV$ & $V$\\
\cmidrule{1-7}
semb. &$99.9940$&$82.8138$&$99.9996$&$99.9992$&$99.9993$&$99.9988$\\
${\overline{g_2}}$&$0.0530$&$0.0530$&$3.41\times10^{-4}$&$0$&$3.44\times10^{-9}$&$4.58\times10^{-9}$\\
\bottomrule
\end{tabular}
\caption[\small{Semblances of signals propagating through layers and equivalent media}]{\small{Approximate values of semblances (in \%) of signals propagating through thin layers and equivalent media for cases $I$--$V$ discussed in the main text. The approximate values of averaged ${\overline{g_2}}$ are also presented.}}
\label{tab:semblance}
\end{table}
\begin{figure}[!htbp]
    \centering
    \includegraphics[width=0.5\textwidth]{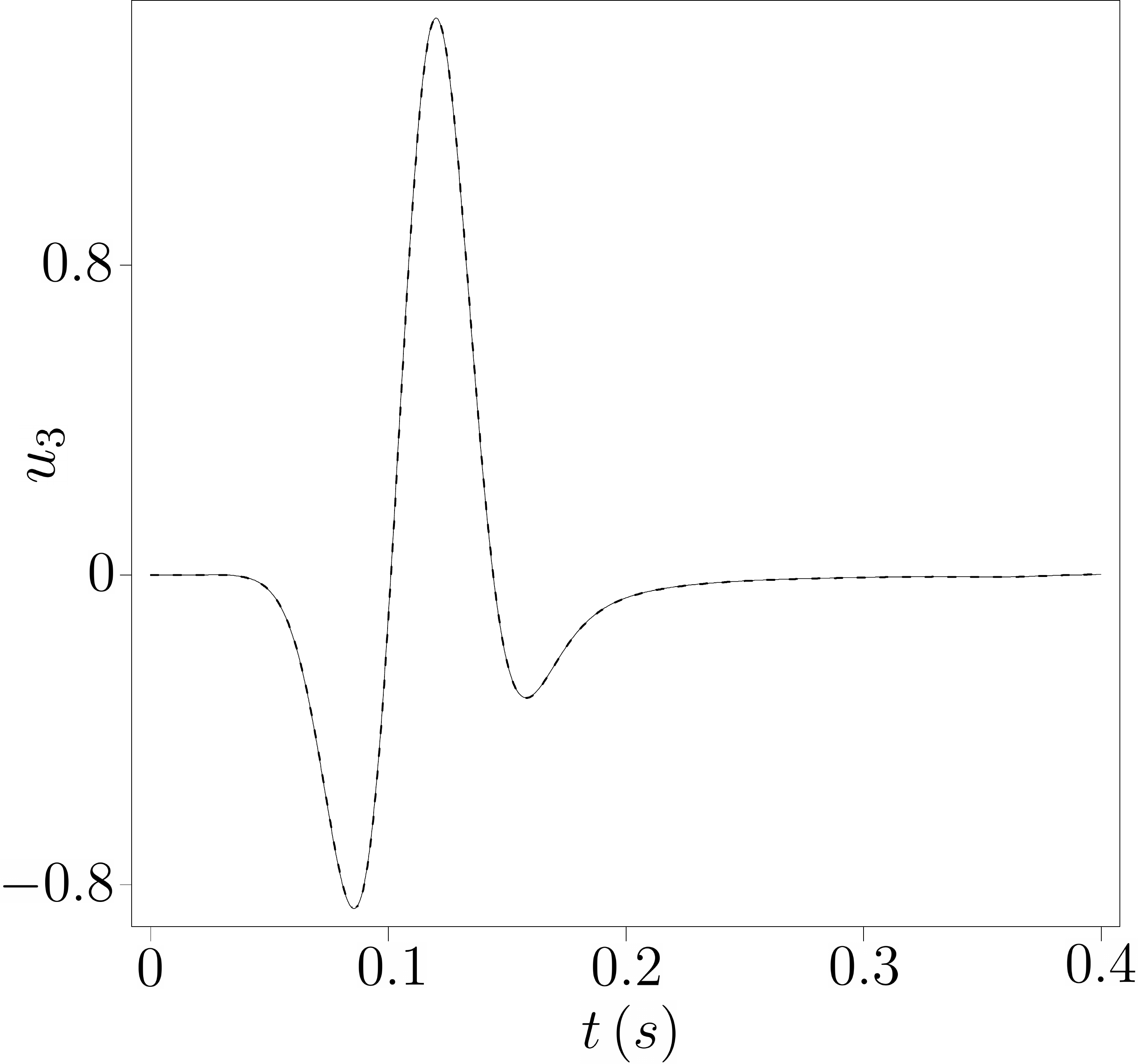}
    \caption[\small{Displacement $u_3$ recorded by the receiver: PL and equivalent medium $IV$}]{\small{Displacement $u_3$ recorded by the receiver. Signal in PL and equivalent medium $IV$ is denoted by dashed and solid line, respectively.}}
    \label{fig:rec4}
\end{figure}
%%%
\begin{figure}[!b]
        \centering
        \begin{subfigure}[b]{0.45\textwidth}
            \centering
            \includegraphics[width=\textwidth]{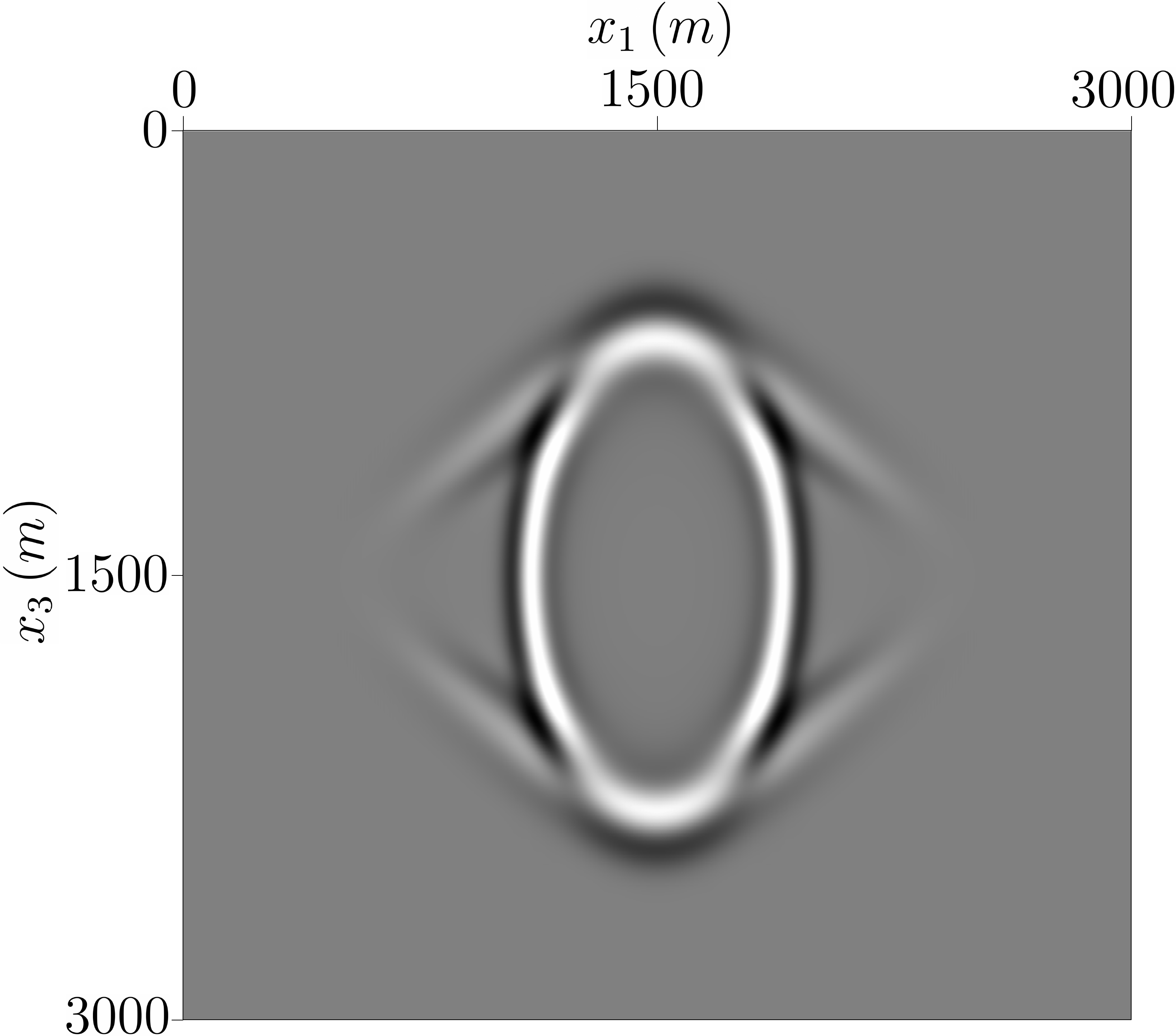}
            \caption%
            {{\footnotesize $u_3$ in PL medium}}    
            \label{fig:disp4a}
        \end{subfigure}
        \quad
        \begin{subfigure}[b]{0.45\textwidth}  
            \centering 
            \includegraphics[width=\textwidth]{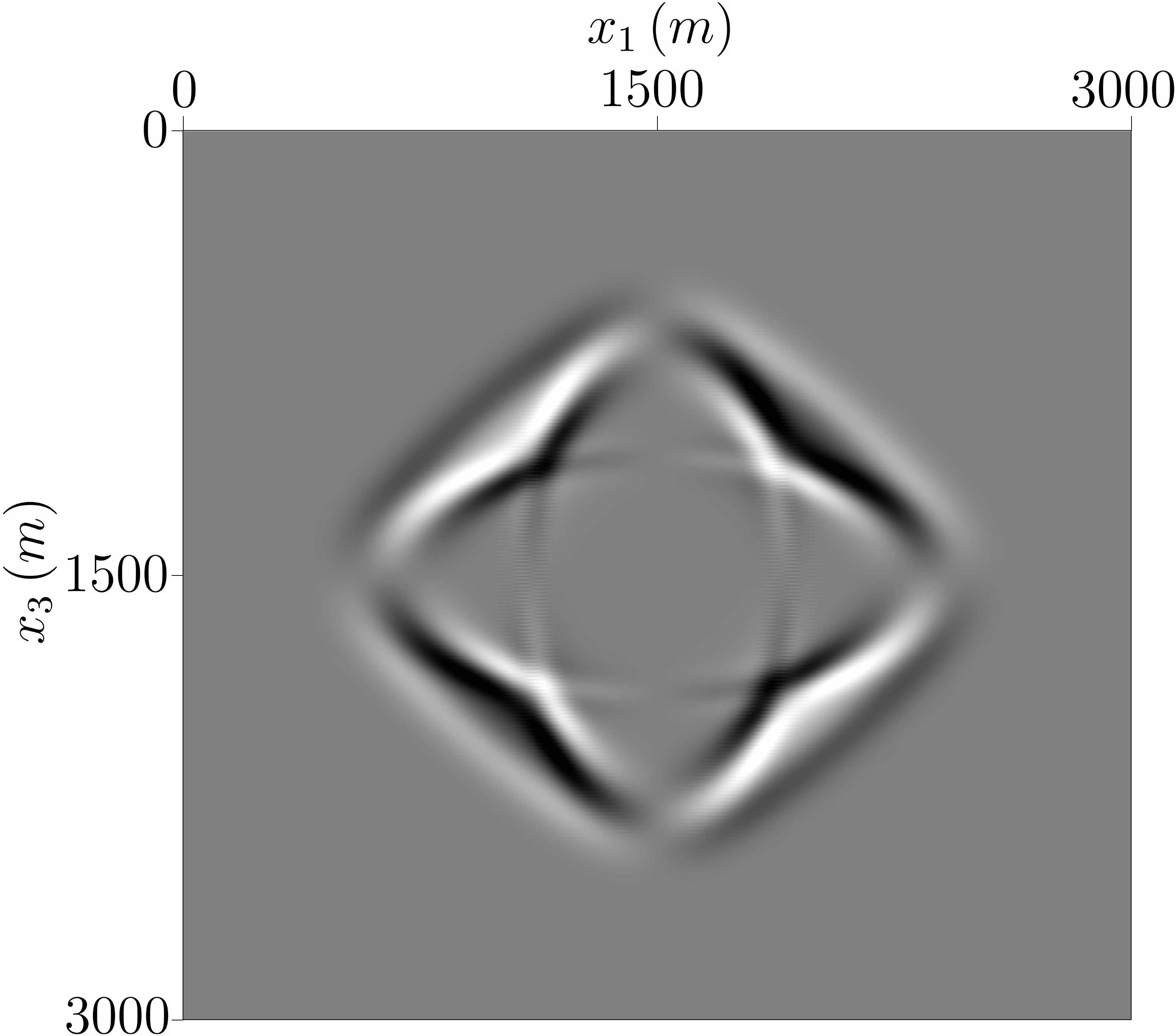}
            \caption[]%
            {{\footnotesize $u_1$ in PL medium}}    
            \label{fig:disp4b}
        \end{subfigure}
        \vskip\baselineskip
        \begin{subfigure}[b]{0.45\textwidth}   
            \centering 
            \includegraphics[width=\textwidth]{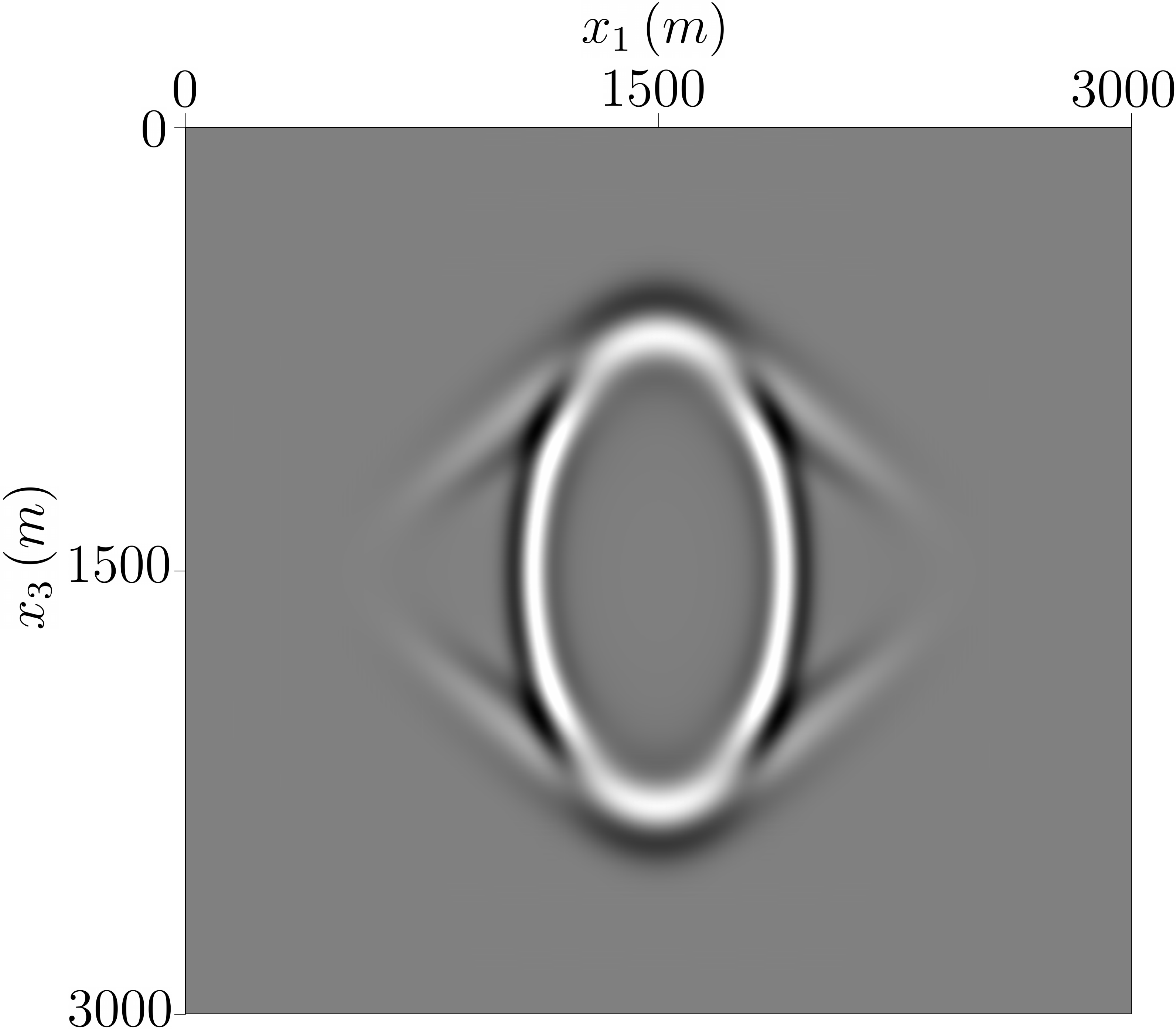}
            \caption[]%
            {{\footnotesize $u_3$ in equivalent medium}}    
            \label{fig:disp4c}
        \end{subfigure}
        \quad
        \begin{subfigure}[b]{0.45\textwidth}   
            \centering 
            \includegraphics[width=\textwidth]{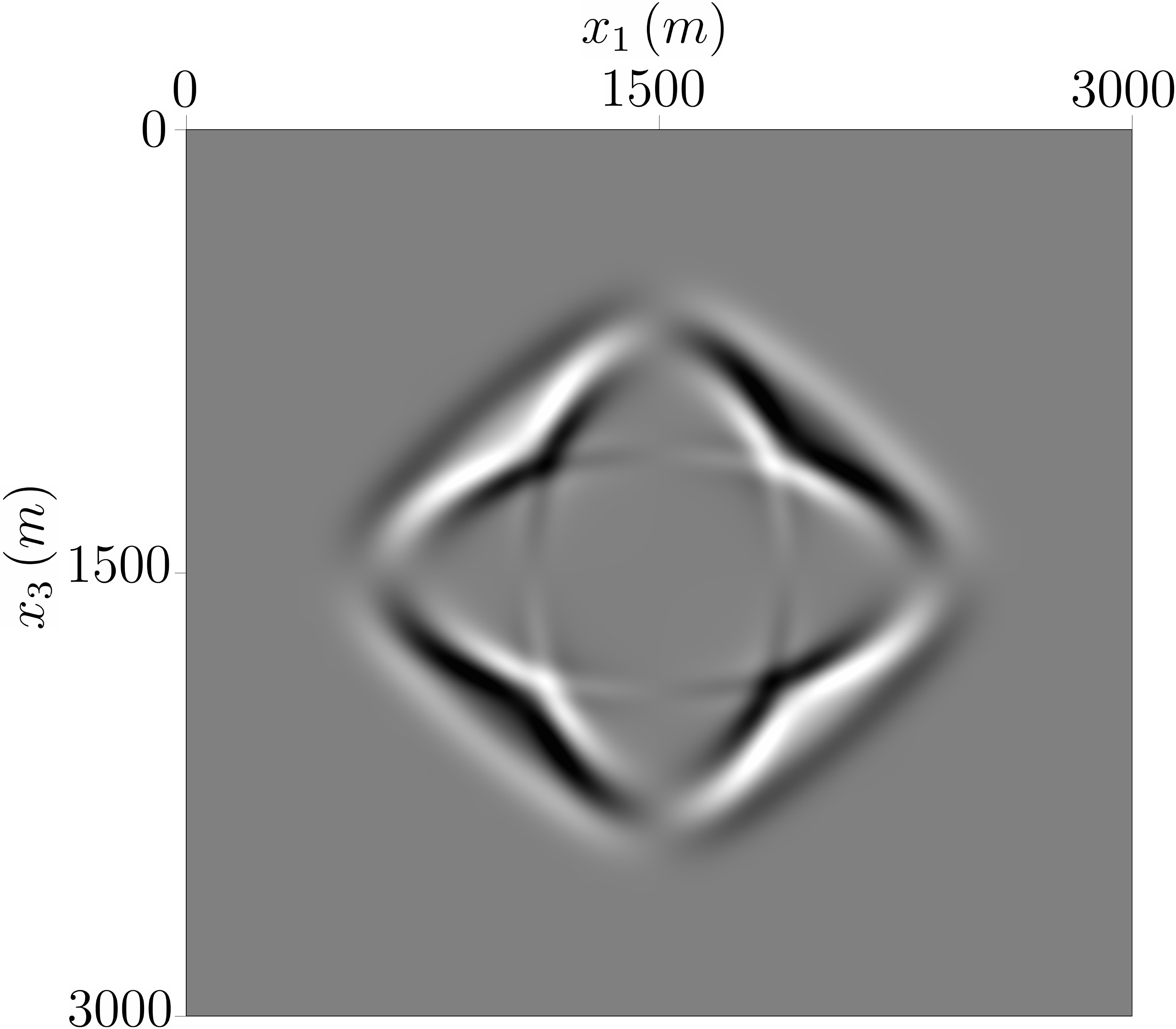}
            \caption[]%
            {{\small \footnotesize $u_1$ in equivalent medium}}    
            \label{fig:disp4d}
        \end{subfigure}
        \caption[\small{Snapshots of displacement $u_3$ and $u_1$\,: PL and equivalent medium $IV$}]
        {\small Snapshots of displacement $u_3$ and $u_1$ in PL and equivalent medium $IV$ at time $t=0.3\,\rm{s}$} 
        \label{fig:disp4}
    \end{figure}
%%%%%%%%%%%
%\begin{comment}
%\end{comment}
%%%%%%%%%%%%%%%%%%%%%
%%%%%%%%%%%%%%%%%%%%%
\newpage
\section{Conclusions}
%%%%%%%%%%%%%%%%%%%%
We focus on the case of product approximation that leads to inaccurate results. We discuss a possibility of its occurrence in physics, in general, and in applied seismology, in particular. We examine numerically the effect of such an inaccuracy on wave propagation in a medium obtained by the Backus average.

In Section~\ref{sec:th}, we present Table~\ref{tab:g} that consists of all the possibilities (up to monoclinic class) of rapidly-varying functions $g$\,.
Table~\ref{tab:gpos} indicates which $g$ may be negative and still obey the stability conditions.
In turn, negative $g$ (or positive, but low values of $g$) in certain layers may lead to the average $\overline{g}\approx0$\,, which makes the product approximation inaccurate.
As discussed in Section~\ref{sec:pois}, for isotropic, cubic, TI, and tetragonal symmetry classes, negative $g$ is tantamount to negative Poisson's ratio in some direction.
Based on the literature review, we show that there are numerous examples in which $\nu<0$ occurs in practice.
Thus, the problematic case of product approximation is likely to occur in real seismological cases, not as assumed previously~\citep{BosProductApprox}.
In general, the chances for negative, or low positive Poisson's ratio are larger if the rock is dry or gas-bearing, is quartz-rich, has numerous cracks and low porosity, occurs in a high-temperature or low-pressure environment.  

In Section~\ref{sec:num}, we perform several 2D numerical simulations of wave propagation in layered and equivalent media with $\overline{g}\approx0$\,.
Based on these examples, we conclude that the problematic case of product approximation that causes the Backus average to be inaccurate does not affect the wave propagation in a meaningful manner.
The product assumption occurs to be much less critical than the long-wave and thin layers assumption.

Please note that our numerical analysis is not entirely complete. 
We neither consider 3D examples, nor the cases of layers exhibiting generally-anisotropic or trigonal symmetry classes.
However, given our simulations, we expect that the influence of $\overline{g}\approx0$ on the wave propagation in equivalent medium obtained by the Backus average should also be marginal in these, low-symmetry or 3D examples. 
%%%%%%%%%%%%%%%%%%%%%
%%%%%%%%%%%%%%%%%%%%%
\section*{Acknowledgements}
%%%%%%%%%%%%%%%%%%
%%%%%%%%%%%%%%%%%%
We wish to thank Michael A. Slawinski and David Dalton for their comments.
The research was done in the context of The Geomechanics Project partially supported by the Natural Sciences and Engineering Research Council of Canada, grant 202259.
The author has no conflict of interests to declare.
%%%%%%%%%%%%%%%%%%
%\section*{Data Availability Statement}
%The data that support the findings of this study are available from the author upon reasonable request.
%%%%%%%%%%%%%%%%%%%%%%
%%%%%%%%%%%%%%%%%%
\bibliographystyle{apa}
\bibliography{bibliography}

\begin{thebibliography}{}

\bibitem[\protect\astroncite{Backus}{1962}]{Backus}
Backus, G.~E. (1962).
\newblock Long-wave elastic anisotropy produced by horizontal layering.
\newblock {\em J. Geophys. Res.}, 67(11):4427--4440.

\bibitem[\protect\astroncite{Bakulin and Grechka}{2003}]{Bakulinetal}
Bakulin, A. and Grechka, V. (2003).
\newblock Effective anisotropy of layered media.
\newblock {\em Geophysics}, 68(6):1817--1821.

\bibitem[\protect\astroncite{Baughman et~al.}{1998}]{Baughman}
Baughman, R.~H., Stafstrom, S., Cui, C., and Dantas, S.~O. (1998).
\newblock Materials with {N}egative {C}ompressibilities in {O}ne or {M}ore
  {D}imensions.
\newblock {\em Science}, 279(5356):1522--1524.

\bibitem[\protect\astroncite{Beljankin and
  Petrov}{1938}]{CristobaliteSedimentary}
Beljankin, D.~S. and Petrov, V.~P. (1938).
\newblock Occurrence of cristobalite in a sedimentary rock.
\newblock {\em Am. Mineral.}, 23(3):153--155.

\bibitem[\protect\astroncite{Bos et~al.}{2017}]{Bosetal}
Bos, L., Dalton, D.~R., Slawinski, M.~A., and Stanoev, T. (2017).
\newblock On {B}ackus {A}verage for {G}enerally {A}nisotropic {L}ayers.
\newblock {\em J. Elast.}, 127(2):179--196.

\bibitem[\protect\astroncite{Bos et~al.}{2018}]{BosProductApprox}
Bos, L., Danek, T., Slawinski, M.~A., and Stanoev, T. (2018).
\newblock Statistical and {}numerical {C}onsiderations of {B}ackus-{A}verage
  {P}roduct {A}pproximation.
\newblock {\em J. Elast.}, 132(1):141--159.

\bibitem[\protect\astroncite{Calvert et~al.}{1977}]{CristobaliteCherts}
Calvert, S.~E., Burns, R.~G., Smith, J.~V., and Kempe, D. R.~C. (1977).
\newblock Mineralogy of {S}ilica {P}hases in {D}eep-{S}ea {C}herts and
  {P}orcelanites [and {D}isucssion].
\newblock {\em Phil. Trans. Roy. Soc. Lond. Math. Phys. Sci.},
  286(1336):239--252.

\bibitem[\protect\astroncite{Carcione et~al.}{1991}]{Carcione}
Carcione, J.~M., Kosloff, D., and Behle, A. (1991).
\newblock Long-wave anisotropy in stratified media: A numerical test.
\newblock {\em Geophysics}, 56(2):245--254.

\bibitem[\protect\astroncite{Castagna and Smith}{1994}]{CastagnaSmith}
Castagna, J.~P. and Smith, S.~W. (1994).
\newblock Comparison of {AVO} indicators: {A} modeling study.
\newblock {\em Geophysics}, 59(12):1849--1855.

\bibitem[\protect\astroncite{Coyner}{1984}]{Coyner}
Coyner, K.~B. (1984).
\newblock {\em Effects of stress, pore pressure, and pore fluids on bulk
  strain, velocity, and permeability in rocks}.
\newblock PhD thesis, Massachusetts {I}nstitute of {T}echnology.

\bibitem[\protect\astroncite{Damby et~al.}{2014}]{CristobaliteVolcanic}
Damby, D.~E., Llewellin, E.~W., Horwell, C.~J., Williamson, B.~J., Najorka, J.,
  Cressey, G., and Carpenter, M. (2014).
\newblock The $\alpha$---$\beta$ phase transition in volcanic cristobalite.
\newblock {\em J. Appl. Crystallogr.}, 47(4):1205--1215.

\bibitem[\protect\astroncite{Dvorkin et~al.}{1999}]{Dvorkin}
Dvorkin, J., Moos, D., Packwood, J.~L., and Nur, A.~M. (1999).
\newblock Identifying patchy saturation from well logs.
\newblock {\em Geophysics}, 64(6):1756--1759.

\bibitem[\protect\astroncite{Dziewo{\'n}ski and Anderson}{1981}]{Dziewonski}
Dziewo{\'n}ski, A.~M. and Anderson, D.~L. (1981).
\newblock Preliminary reference {E}arth model.
\newblock {\em Phys. Earth Planet. Inter.}, 25(4):297--356.

\bibitem[\protect\astroncite{Emery and Stewart}{2006}]{Crewes}
Emery, D.~J. and Stewart, R.~R. (2006).
\newblock Using {VP}/{VS} to explore for sandstone reservoirs: well log and
  synthetic seismograms from the {J}eanne d'{A}rc basin, offshore
  {N}ewfoundland.
\newblock {\em CREWES Research Report}, 18(1):1--20.

\bibitem[\protect\astroncite{Fay et~al.}{2012}]{FayEtAl2012}
Fay, M., Larter, S., Bennett, B., Snowdon, L., and Fowler, M. (2012).
\newblock {P}etroleum {S}ystems and {Q}uantitative {O}il {C}hemometric {M}odels
  for {H}eavy {O}ils in {A}lberta and {S}askatchewan to {C}haracterize
  {O}il-{S}ource {C}orrelations.
\newblock {\em GeoConvention 2012:Vision}.

\bibitem[\protect\astroncite{Fomel et~al.}{2013}]{Fomel}
Fomel, S., Sava, P., Vlad, I., .Liu, Y., and Bashkardin, V. (2013).
\newblock Madagascar: {O}pen-source software project for multidimensional data
  analysis and reproducible computational experiments.
\newblock {\em J. Open Res. Softw.}, 1(1):e8.

\bibitem[\protect\astroncite{Freund}{1992}]{Freund}
Freund, D. (1992).
\newblock Ultrasonic compressional and shear velocities in dry clastic rocks as
  a function of porosity, clay content, and confining pressure.
\newblock {\em Geophys. J. Int.}, 108(1):125--135.

\bibitem[\protect\astroncite{Goodway}{2001}]{Goodway}
Goodway, B. (2001).
\newblock A{VO} and {L}am\'e constants for rock parameterization and fluid
  detection.
\newblock {\em CSEG Recorder}, 26(6):39--60.

\bibitem[\protect\astroncite{Gregory}{1976}]{Gregory}
Gregory, A.~R. (1976).
\newblock Fluid saturation effect on dynamic elastic properties of sedimentary
  rocks.
\newblock {\em Geophysics}, 41(5):895--921.

\bibitem[\protect\astroncite{Grima et~al.}{2007}]{Zeolite2}
Grima, J.~N., Gatt, R., Zammit, V., Williams, J.~J., Evans, K.~E., Alderson,
  A., and Walton, R.~I. (2007).
\newblock Natrolite: a zeolite with negative {P}oisson's ratios.
\newblock {\em J. Appl. Phys.}, 101(8):086102.

\bibitem[\protect\astroncite{Grima et~al.}{2000}]{Zeolite}
Grima, J.~N., Jackson, R., Alderson, A., and Evans, K.~E. (2000).
\newblock Do {Z}eolites {H}ave {N}egative {P}oisson's {R}atios?
\newblock {\em Advanced Materials}, 12(24):1912--1918.

\bibitem[\protect\astroncite{Han}{1986}]{Han}
Han, D.-H. (1986).
\newblock {\em Effects of porosity and clay content on acoustic properties of
  sandstones and unconsolidated sediments:}.
\newblock PhD thesis, Stanford University.

\bibitem[\protect\astroncite{Hay}{1986}]{ZeoliteOccur}
Hay, R.~L. (1986).
\newblock Geologic {O}ccurrence of {Z}eolites and {S}ome {A}ssociated
  {M}inerals.
\newblock {\em Stud. Surf. Sci. Catal.}, 28(1):35--40.

\bibitem[\protect\astroncite{Hayes et~al.}{1994}]{HayesEtAl1994}
Hayes, B. J.~R., Christopher, J.~R., Rosenthal, L., Los, G., McKercher, B.,
  Minken, D., Tremblay, Y.~M., and Fennel, J. (1994).
\newblock Cretaceous {M}annville {G}roup of the {W}estern {C}anada
  {S}edimentary {B}asin; in {G}eological {A}tlas of the {W}estern {C}anada
  {S}edimentary {B}asin.
\newblock {\em Canadian Society of Petroleum Geologists and Alberta Research
  Council}, 19(URL:
  https://ags.aer.ca/publications/chapter-19-cretaceous-mannville-group.htm).

\bibitem[\protect\astroncite{Helbig}{1984}]{Helbig84}
Helbig, K. (1984).
\newblock Anisotropy and dispersion in periodically layered media.
\newblock {\em Geophysics}, 49(4):364--373.

\bibitem[\protect\astroncite{Helbig}{1994}]{Helbig}
Helbig, K. (1994).
\newblock {\em Foundations of anisotropy for exploration seismics}.
\newblock Pergamon Press, 1st edition.

\bibitem[\protect\astroncite{Hommand-Etienne and Houpert}{1989}]{Houpert}
Hommand-Etienne, F. and Houpert, R. (1989).
\newblock Thermally {I}nduced {M}icrocracking in {G}ranites: {C}haracterization
  and {A}nalysis.
\newblock {\em Int. J. Rock Mech. Min. Sci. Geomech. Abstr.}, 26(2):125--134.

\bibitem[\protect\astroncite{Ji et~al.}{2018}]{JiEtAl2018}
Ji, S., Li, L., Motra, H.~B., Wuttke, F., Sun, S., Michibayashi, K., and
  Salisbury, M.~H. (2018).
\newblock Poisson's {R}atio and {A}uxetic {P}roperties of {N}atural {R}ocks.
\newblock {\em J. Geophys. Res. Solid Earth}, 123(2):1161--1185.

\bibitem[\protect\astroncite{Ji et~al.}{2010}]{JiEtAl2010}
Ji, S., Sun, S., Wang, Q., and Marcotte, D. (2010).
\newblock Lam{\'e} parameters of common rocks in the earth's crust and upper
  mantle.
\newblock {\em J. Geophys. Res. Solid Earth}, 115(B6).

\bibitem[\protect\astroncite{Ji et~al.}{2019}]{JiEtAl2019}
Ji, S., Wang, Q., and Li, L. (2019).
\newblock Seismic velocities, {P}oisson's ratios and potential auxetic behavior
  of volcanic rocks.
\newblock {\em Tectonophysics}, 766(1):270--282.

\bibitem[\protect\astroncite{Jizba}{1991}]{Jizba}
Jizba, D. (1991).
\newblock {\em Mechanical and acoustical properties of sandstones and shales}.
\newblock PhD thesis, Stanford University.

\bibitem[\protect\astroncite{Kudela and Stanoev}{2018}]{KudelaStanoev}
Kudela, I. and Stanoev, T. (2018).
\newblock On possibile issues of {B}ackus average.
\newblock {\em arXiv:1804.01917 [physics.geo-ph]}.

\bibitem[\protect\astroncite{Lakes}{2017}]{Lakes}
Lakes, R.~S. (2017).
\newblock Negative-{P}oisson's-{R}atio {M}aterials: {A}xetic {S}olids.
\newblock {\em Annu. Rev. Mater. Res.}, 47(1):63--81.

\bibitem[\protect\astroncite{Liner and Fei}{2007}]{LinerFei}
Liner, C. and Fei, T. (2007).
\newblock The backus number.
\newblock {\em Lead. Edge}, 26(4):420--426.

\bibitem[\protect\astroncite{Mavko and Jizba}{1994}]{MavkoJizba}
Mavko, G. and Jizba, D. (1994).
\newblock The relation between seismic {P}- and {S}-wave velocity dispersion in
  saturated rocks.
\newblock {\em Geophysics}, 59(1):87--92.

\bibitem[\protect\astroncite{Mavko et~al.}{2009}]{MavkoEtAl2009}
Mavko, G., Mukerji, T., and Dvorkin, J. (2009).
\newblock {\em The rock physics handbook}.
\newblock Cambridge, 2nd edition.

\bibitem[\protect\astroncite{McKnight et~al.}{2008}]{Quartz}
McKnight, R. E.~A., Moxon, T., Buckley, A., Taylor, P.~A., Darling, T.~W., and
  Carpenter, M.~A. (2008).
\newblock Grain size dependence of elastic anomalies accompanying the
  $\alpha$--$\beta$ phase transition in polycrystalline quartz.
\newblock {\em J. Phys. Condens. Matter}, 20(7):075229.

\bibitem[\protect\astroncite{Mizota et~al.}{1987}]{CristobaliteSoil}
Mizota, C., Toh, N., and Matsuhisa, Y. (1987).
\newblock Origin of cristobalite in soils derived from volcanic ash in
  temperate and tropical regions.
\newblock {\em Geoderma}, 39(4):323--330.

\bibitem[\protect\astroncite{Mouchat and Coudert}{2014}]{Mouchat}
Mouchat, F. and Coudert, F.~X. (2014).
\newblock Necessary and {S}ufficient {E}lastic {S}tability {C}onditions in
  {V}arious {C}rystal {S}ystems.
\newblock {\em Phys. Rev. B}, 90(22):1--4.

\bibitem[\protect\astroncite{Nur and Simmons}{1969}]{Nur}
Nur, A. and Simmons, G. (1969).
\newblock The effect of saturation on velocity in low porosity rocks.
\newblock {\em Earth Planet. Sci. Lett.}, 7(2):183--193.

\bibitem[\protect\astroncite{Sanchez-Valle et~al.}{2008}]{Zeolite3}
Sanchez-Valle, C., Lethbridge, Z. A.~D., Sinogeikin, S.~V., Williams, J.~J.,
  Walton, R.~I., Evans, K.~E., and Bass, J.~D. (2008).
\newblock Negative {P}oisson's ratios in siliceous zeolite {MFI}-silicalite.
\newblock {\em J. Chem. Phys.}, 128(18):184503.

\bibitem[\protect\astroncite{Schoenberg and Muir}{1989}]{SchMuir}
Schoenberg, M. and Muir, F. (1989).
\newblock A calculus for finely layered anisotropic media.
\newblock {\em Geophysics}, 54(5):581--589.

\bibitem[\protect\astroncite{Slawinski}{2015}]{SlawinskiRed}
Slawinski, M.~A. (2015).
\newblock {\em Waves and rays in elastic continua}.
\newblock World Scientific, 3rd edition.

\bibitem[\protect\astroncite{Slawinski}{2018}]{SlawinskiGreen}
Slawinski, M.~A. (2018).
\newblock {\em Waves and rays in seismology: {A}nswers to unasked questions}.
\newblock World Scientific, 2nd edition.

\bibitem[\protect\astroncite{Yeganeh-Haeri et~al.}{1992}]{Cristobalite}
Yeganeh-Haeri, A., Weidner, D.~J., and Parise, J.~B. (1992).
\newblock Elasticity of $\alpha$-{C}ristobalite: {A} {S}ilicon {D}ioxide with a
  {N}egative {P}oisson's {R}atio.
\newblock {\em Science}, 257(5070):650--652.

\bibitem[\protect\astroncite{Zaitsev et~al.}{2017}]{Zaitsev}
Zaitsev, V.~Y., Radostin, A.~V., Pasternak, E., and Dyskin, A. (2017).
\newblock Extracting real-crack properties from non-linear elastic behaviour of
  rocks: abundance of cracks with dominating normal compliance and rocks with
  negative {P}oisson ratios.
\newblock {\em Nonlinear Process. Geophys.}, 24(3):543--551.

\end{thebibliography}
%%%%%%%%%%%%%%%%%%%%%%%%%%%%%%%%%%%%%%%%%%%%%%%%%%%%%%%%%%%%%%%%%%%%%%%%%%%%%%%%%%
\appendix
\addcontentsline{toc}{section}{Appendices}
%%%%%%%%
\section{Backus average for anisotropic layers}\label{ch4:ap1}
%%%%%%%%%%
Let us write the strain-stress relations in two dimensions ($x_3x_1$-plane), namely,
\begin{equation}\label{eq:ap1}
\sigma_{11}=C_{11}\varepsilon_{11}+C_{13}\varepsilon_{33}\,,
\end{equation}
\begin{equation}
\sigma_{33}=C_{13}\varepsilon_{11}+C_{33}\varepsilon_{33}\,,
\end{equation}
\begin{equation}\label{eq:ap3}
\sigma_{13}=2C_{55}\varepsilon_{13}\,,
\end{equation}
which are the relations valid for the monoclinic, orthotropic, tetragonal, and TI symmetry class.
Upon a rearrangement, we get
\begin{equation}
\sigma_{11}=\left(C_{11}-\frac{C_{13}^2}{C_{33}}\right)\varepsilon_{11}+\left(\frac{C_{13}}{C_{33}}\right)\sigma_{33}\,,
\end{equation}
\begin{equation}
\varepsilon_{33}=-\left(\frac{C_{13}}{C_{33}}\right)\varepsilon_{11}
+\left(\frac{1}{C_{33}}\right)\sigma_{33}\,,
\end{equation}
\begin{equation}
\frac{\partial{u_1}}{\partial{x_3}}=\left(\frac{1}{C_{55}}\right)\sigma_{13}-\frac{\partial{u_3}}{\partial{x_1}}\,.
\end{equation}
Let us treat the above equations as the stress-strain relations that correspond to many individual constituents that we want to average.
To perform the averaging process, we use the three following properties:
the average of the sum is a sum of the average, the average of the derivative is a derivative of the average, and, finally, the product approximation.
We obtain
\begin{equation}\label{eq:ap4}
\overline{\sigma_{11}}=\left[\overline{\left(C_{11}-\frac{C_{13}^2}{C_{33}}\right)}+\overline{\left(\frac{C_{13}}{C_{33}}\right)}^2\overline{\left(\frac{1}{C_{33}}\right)}^{-1}\right]\overline{\varepsilon_{11}}+\overline{\left(\frac{C_{13}}{C_{33}}\right)}\,\overline{\left(\frac{1}{C_{33}}\right)}^{-1} \overline{\varepsilon_{33}}\,,
\end{equation}
\begin{equation}
\overline{\sigma_{33}}=\overline{\left(\frac{C_{13}}{C_{33}}\right)}\,\overline{\left(\frac{1}{C_{33}}\right)}^{-1}\overline{\varepsilon_{11}}+\overline{\left(\frac{1}{C_{33}}\right)}^{-1} \overline{\varepsilon_{33}}\,,
\end{equation}
\begin{equation}\label{eq:ap6}
\overline{\sigma_{13}}=\overline{\left(\frac{1}{C_{55}}\right)}^{-1}2\,\overline{\varepsilon_{13}}\,.
\end{equation}
Comparing equations~(\ref{eq:ap4})--(\ref{eq:ap6}) with equations~(\ref{eq:ap1})--(\ref{eq:ap3}), we see that the equivalent elasticity parameters are equal to
\begin{equation}
C_{11}^{\rm {\overline{eq}}}=\overline{\left(C_{11}-\frac{C_{13}^2}{C_{33}}\right)}+\overline{\left(\frac{C_{13}}{C_{33}}\right)}^2\overline{\left(\frac{1}{C_{33}}\right)}^{-1}\,,
\end{equation}
\begin{equation}
C_{13}^{\rm {\overline{eq}}}=\overline{\left(\frac{C_{13}}{C_{33}}\right)}\,\overline{\left(\frac{1}{C_{33}}\right)}^{-1} \,,
\end{equation}
\begin{equation}
C_{33}^{\rm {\overline{eq}}}=\overline{\left(\frac{1}{C_{33}}\right)}^{-1}\,,
\end{equation}
\begin{equation}
C_{55}^{\rm {\overline{eq}}}=\overline{\left(\frac{1}{C_{55}}\right)}^{-1}\,,
\end{equation}
and the resulting medium is either monoclinic, orthotropic, tetragonal, or TI.
If layers have cubic symmetry, then $C_{33}=C_{11}\,$.
In such a case, $C_{33}^{\rm{\overline{eq}}}\neq C_{11}^{\rm{\overline{eq}}}$\,, which means that the equivalent medium is not cubic.
To understand what is the symmetry class of the medium equivalent to cubic layers, we need to derive the analogous equivalent parameters, but for 3D case.
Upon an analogous procedure, shown above, we get
 \begin{equation}
C_{11}^{\rm {\overline{eq}}}=\overline{\left(C_{11}-\frac{C_{13}^2}{C_{11}}\right)}+\overline{\left(\frac{C_{13}}{C_{11}}\right)}^2\overline{\left(\frac{1}{C_{11}}\right)}^{-1}\,,
\end{equation}
\begin{equation}
C_{12}^{\rm {\overline{eq}}}=\overline{\left(C_{13}-\frac{C_{13}^2}{C_{11}}\right)}+\overline{\left(\frac{C_{13}}{C_{11}}\right)}^2\overline{\left(\frac{1}{C_{11}}\right)}^{-1}\,,
\end{equation}
\begin{equation}
C_{13}^{\rm {\overline{eq}}}=\overline{\left(\frac{C_{13}}{C_{11}}\right)}\,\overline{\left(\frac{1}{C_{11}}\right)}^{-1} \,,
\end{equation}
\begin{equation}
C_{33}^{\rm {\overline{eq}}}=\overline{\left(\frac{1}{C_{11}}\right)}^{-1}\,,
\end{equation}
\begin{equation}
C_{55}^{\rm {\overline{eq}}}=\overline{\left(\frac{1}{C_{55}}\right)}^{-1}\,,
\end{equation}
\begin{equation}
C_{66}^{\rm {\overline{eq}}}=\overline{C_{55}}\,,
\end{equation}
where $C_{11}^{\rm {\overline{eq}}}=C_{22}^{\rm {\overline{eq}}}$\,, $C_{13}^{\rm {\overline{eq}}}=C_{23}^{\rm {\overline{eq}}}$\,, and $C_{55}^{\rm {\overline{eq}}}=C_{44}^{\rm {\overline{eq}}}$\,.
The equivalent medium has six independent elasticity parameters and exhibits the tetragonal symmetry class.

%%%%%%%%
\section{Backus procedure for a trigonal tensor}\label{sec:appendix}
%%%%%%%%%%
First, we write the stress-strain relations in a trigonal medium (expressed in a natural coordinate system) as
\begin{equation}\label{eq1b}
\sigma_{11}=C_{11}\varepsilon_{11}+C_{12}\varepsilon_{22}+C_{13}\varepsilon_{33}+C_{15}\frac{\partial u_1}{\partial x_3}+C_{15}\frac{\partial u_3}{\partial x_1}\,,
\end{equation}
\begin{equation}\label{eq2b}
\sigma_{22}=C_{12}\varepsilon_{11}+C_{11}\varepsilon_{22}+C_{13}\varepsilon_{33}-C_{15}\frac{\partial u_1}{\partial x_3}-C_{15}\frac{\partial u_3}{\partial x_1}\,,
\end{equation}
\begin{equation}\label{eq3b}
\sigma_{33}=C_{13}\varepsilon_{11}+C_{13}\varepsilon_{22}+C_{33}\varepsilon_{33}\,,
\end{equation}
\begin{equation}\label{eq4b}
\sigma_{23}=C_{44}\frac{\partial{u_2}}{\partial{x_3}}+C_{44}\frac{\partial{u_3}}{\partial{x_2}}-2C_{15}\varepsilon_{12}\,,
\end{equation}
\begin{equation}\label{eq5b}
\sigma_{13}=C_{44}\frac{\partial{u_1}}{\partial{x_3}}+C_{44}\frac{\partial{u_3}}{\partial{x_1}}+C_{15}\varepsilon_{11}-C_{15}\varepsilon_{22}\,,
\end{equation}
\begin{equation}\label{eq6b}
\sigma_{12}=(C_{11}-C_{12})\varepsilon_{12}-C_{15}\frac{\partial{u_2}}{\partial{x_3}}-C_{15}\frac{\partial{u_3}}{\partial{x_2}}\,.
\end{equation}
We can directly rewrite equations~(\ref{eq3b})--(\ref{eq5b}) in a manner that the nearly-constant stresses and strains are on the right-hand side, whereas the sole varying function of displacements is on the left-hand side. 
We get,
\begin{equation}\label{eq7b}
\varepsilon_{33}=\sigma_{33}\underbrace{ \left(\frac{1}{C_{33}}\right) }_{\text{$g_1$}}-\underbrace{\left(\frac{C_{13}}{C_{33}}\right)}_\text{$g_2$}\varepsilon_{11}-\underbrace{\left(\frac{C_{13}}{C_{33}}\right)}_\text{$g_3$}\varepsilon_{22}\,,
\end{equation}
\begin{equation}\label{eq8b}
\frac{\partial{u_2}}{\partial{x_3}}=\sigma_{23}\underbrace{\left(\frac{1}{C_{44}}\right)}_\text{$g_4$}-\frac{\partial{u_3}}{\partial{x_2}}-\underbrace{\left(\frac{C_{15}}{C_{44}}\right)}_\text{$g_{t}$}2\varepsilon_{12}\,,
\end{equation}
\begin{equation}\label{eq9b}
\frac{\partial{u_1}}{\partial{x_3}}=\sigma_{13}\underbrace{\left(\frac{1}{C_{44}}\right)}_\text{$g_5$}-\frac{\partial{u_3}}{\partial{x_1}}-\left(\frac{C_{15}}{C_{44}}\right)\varepsilon_{11}+\left(\frac{C_{15}}{C_{44}}\right)\varepsilon_{22}\,.
\end{equation}
Now, we insert the right-hand side of equation~(\ref{eq7b}) and (\ref{eq9b}) into equations~(\ref{eq1b}) and (\ref{eq2b}).
Also, we insert the right-hand side of~(\ref{eq8b}) into~(\ref{eq6b}).
Upon simple calculations, we obtain
\begin{equation}
\sigma_{11}= \sigma_{33}\left(\frac{C_{13}}{C_{33}}\right)
+
\sigma_{13}\left(\frac{C_{15}}{C_{44}}\right)
+
\underbrace{\left(C_{11}-\frac{C_{13}^2}{C_{33}}-\frac{C_{15}^2}{C_{44}}\right)}_\text{$g_6$}\varepsilon_{11}
+
\underbrace{\left(C_{12}-\frac{C_{13}^2}{C_{33}}+\frac{C_{15}^2}{C_{44}}\right)}_\text{$g_7$}\varepsilon_{22}
\,,
\end{equation}
\begin{equation}
\sigma_{22}=
 \sigma_{33}\left(\frac{C_{13}}{C_{33}}\right)
-
\sigma_{13}\left(\frac{C_{15}}{C_{44}}\right)
+
\left(C_{12}-\frac{C_{13}^2}{C_{33}}+\frac{C_{15}^2}{C_{44}}\right)\varepsilon_{11}
+
\underbrace{\left(C_{11}-\frac{C_{13}^2}{C_{33}}-\frac{C_{15}^2}{C_{44}}\right)}_\text{$g_8$}\varepsilon_{22}
\,,
\end{equation}
\begin{equation}\label{eq12b}
\sigma_{12}=
-\sigma_{23}\left(\frac{C_{15}}{C_{44}}\right)
-
\underbrace{\left(\frac{C_{11}-C_{12}}{2}-\frac{C_{15}^2}{C_{44}}\right)}_\text{$g_9$}2\varepsilon_{12}
\,.
\end{equation}
Terms in parenthesis in equations~(\ref{eq7b})--(\ref{eq12b}) correspond to various $g$\,; we denote them as $g_i$ or $g_{t}$\,. 
We notice that in case of trigonal symmetry, $g_2=g_3$\,, $g_4=g_5$\,, and $g_6=g_8$\,.
%%%%%%%%
\section{Relation between $n_1$\,, $n_2$\,, $g^{\rm ort}_2$\,, and $g^{\rm ort}_3$}\label{ap2}
%%%%%%%%%%
\begin{lemma}
If numerators of Poisson's ratios $n_1>0$ and $n_2>0$\,, then the stability conditions for orthotropic media do not allow $g_2^{\rm ort}<0$ and $g_3^{\rm ort}<0$\,.
\end{lemma}
\begin{proof}
Consider $n_1>0$\,, namely,
\begin{equation}\label{ap:n1}
C_{13}C_{22}-C_{12}C_{23}>0\,.
\end{equation}
Let us assume that
$g_2^{\rm ort}<0$ and $g_3^{\rm ort}<0$\,.
Since, according to stability conditions, $C_{33}\geq0$\,, the above assumption is tantamount to $C_{13}<0$ and $C_{23}<0\,$.
We also know that $C_{22}\geq0$\,, thus, to satisfy expression~(\ref{ap:n1}) $C_{12}$ must be positive. 
Therefore, we can write
\begin{equation}\label{ap:n1b}
\frac{C_{13}C_{22}}{C_{12}}>C_{23}\,.
\end{equation}
Consider $n_2>0$\,, namely,
\begin{equation}\label{ap:n1}
C_{23}C_{11}>C_{12}C_{13}\,.
\end{equation}
Both sides are negative, where $C_{11}$ and $C_{12}$ must be positive.
This inequality allows us to insert some larger value in the place of $C_{23}$\,.
If we insert the left-hand side of inequality~(\ref{ap:n1b}), we get
\begin{equation}
\frac{C_{13}C_{11}C_{22}}{C_{12}}>C_{12}C_{13}\,.
\end{equation}
Since $C_{13}$ is assumed to be negative and $C_{12}$ must be negative, we obtain
\begin{equation}
C_{11}C_{22}<C_{12}^2\,,
\end{equation}
which is not allowed by the stability condition.
\end{proof}

\begin{lemma}
If numerators of Poisson's ratios $n_1<0$ and $n_2<0$\,, then the stability conditions for orthotropic media do not allow $g_2^{\rm ort}>0$ and $g_3^{\rm ort}>0$\,.
\end{lemma}
\begin{proof}
Consider $n_1<0$\,, we get
\begin{equation}\label{ap:n1b}
\frac{C_{12}C_{23}}{C_{22}}>C_{13}\,.
\end{equation}
Also, consider $n_2<0$\,, namely,
\begin{equation}\label{ap:n1c}
C_{12}C_{13}>C_{23}C_{11}\,.
\end{equation}
Let us assume that
$g_2^{\rm ort}>0$ and $g_3^{\rm ort}>0$\,.
In such a case, both sides of inequalities~(\ref{ap:n1b}) and~(\ref{ap:n1c}) are positive. 
We can insert greater value than $C_{13}$ in the inequality~(\ref{ap:n1c}).
We obtain,
\begin{equation}
\frac{C_{12}^2C_{23}}{C_{22}}>C_{23}C_{11}\,.
\end{equation}
Since, $C_{23}>0$\,, we obtain
\begin{equation}
C_{12}^2>C_{11}C_{22}\,,
\end{equation}
which is not allowed by the stability condition.
\end{proof}
Similar strategy can be used to prove that if $n_1<0$ and $n_2>0$ then $g_2^{\rm ort}>0$ and $g_3^{\rm ort}<0$ are not allowed, or, conversely, if $n_1>0$ and $n_2<0$ then $g_2^{\rm ort}<0$ and $g_3^{\rm ort}>0$ are not allowed.

\end{document}